\documentclass[a4paper]{article}

\usepackage[vlined,boxed,ruled,algosection]{algorithm2e} 
\usepackage{amsmath}
\usepackage{latexsym}
\usepackage{amssymb}
\usepackage{graphicx}
\usepackage{amsthm}

\newtheorem{thm}{Theorem}[section]
\newtheorem{lem}[thm]{Lemma}

\newtheorem{pro}[thm]{Proposition}

\theoremstyle{definition}
\newtheorem{defn}{Definition}[section]

\begin{document}

\title{The Longest $(s, t)$-paths of $O$-shaped Supergrid Graphs
}

\author{\vspace{0.5cm}Ruo-Wei Hung$^\textrm{a}$ and Fatemeh Keshavarz-Kohjerdi$^{\mathrm{b},}$\thanks{Corresponding author.}\\
$^\textrm{a}$\textit{Department of Computer Science \& Information Engineering,}\\
\textit{Chaoyang University of Technology, Wufeng, Taichung 41349, Taiwan}\\
\textit{\vspace{0.2cm}e-mail address: rwhung@cyut.edu.tw}\\
$^\textrm{b}$\textit{Department of Mathematics \& Computer Science,}\\
\textit{Shahed University, Tehran, Iran}\\
\textit{e-mail address: f.keshavarz@shahed.ac.ir}}


\maketitle

\begin{abstract}
In this paper, we continue the study of the Hamiltonian and longest $(s, t)$-paths of supergrid graphs. The Hamiltonian $(s, t)$-path of a graph is a Hamiltonian path between any two given vertices $s$ and $t$ in the graph, and the longest $(s, t)$-path is a simple path with the maximum number of vertices from $s$ to $t$ in the graph. A graph holds Hamiltonian connected property if it contains a Hamiltonian $(s, t)$-path. These two problems are well-known NP-complete for general supergrid graphs. An $O$-shaped supergrid graph is a special kind of a rectangular grid graph with a rectangular hole. In this paper, we first prove the Hamiltonian connectivity of $O$-shaped supergrid graphs except few conditions. We then show that the longest $(s, t)$-path of an $O$-shaped supergrid graph can be computed in linear time. The Hamiltonian and longest $(s, t)$-paths of $O$-shaped supergrid graphs can be applied to compute the minimum trace of computerized embroidery machine and 3D printer when a hollow object is printed.

\vspace{0.2cm}\noindent\textbf{Keywords:}
Longest path, Hamiltonian connectivity, Supergrid graphs, $O$-shaped supergrid graphs, Computerized embroidery machines, 3D printers
\end{abstract}

\section{Introduction}\label{Introduction}
The studied graphs, namely \textit{supergrid graphs}, are derived from our industry-university cooperative research project. They can be applied to the computerized embroidery machines. The flow of a computerized sewing process is as follows. Given by a colour image. The computerized embroidery software first uses the image processing technique to produce $k$ blocks of different colors. Then, it computes the stitching trace for each block of colors. Finally, the software transmits its computed stitching trace to computerized embroidery machine, and the machine performs the sewing action along its received stitching trace. 
Since each stitch position of a sewing machine can be moved to its eight neighbor positions (left, right, up, down, up-left, up-right, down-left, and down-right), we define the supergrid graph as follows: Each lattice of a block of color is represented by a vertex and each vertex $v$ is coordinated as $(v_x, v_y)$, denoted by $v = (v_x, v_y)$, where $v_x$ and $v_y$ are integers and represent the $x$ and $y$ coordinates of node $v$, respectively. Two vertices $u$ and $v$ are adjacent if and only if $|u_x-v_x|\leqslant 1$ and $|u_y-v_y|\leqslant 1$. Thus, the possible adjacent vertices of a vertex $v = (v_x, v_y)$ in a supergrid graph contain $(v_x, v_y-1)$, $(v_x-1, v_y)$, $(v_x+1, v_y)$, $(v_x, v_y+1)$, $(v_x-1, v_y-1)$, $(v_x+1, v_y+1)$, $(v_x+1, v_y-1)$, and $(v_x-1, v_y+1)$. 


\begin{figure}[!t]
\begin{center}
\includegraphics[width=0.8\textwidth]{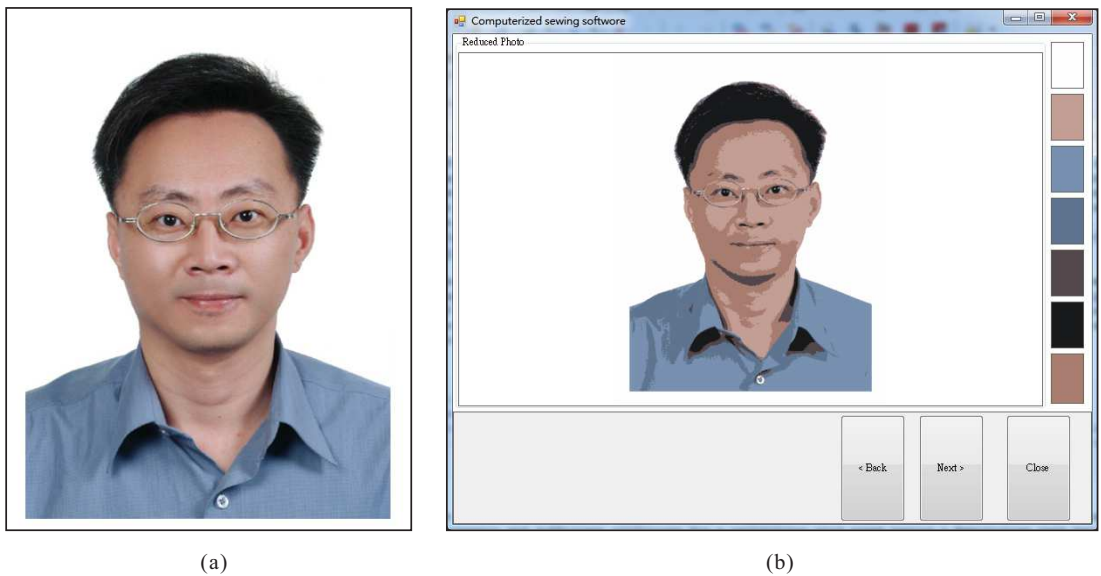}
\caption{(a) An input image for the computerized sewing software and (b) seven colors of regions produced by image processing software.} \label{Fig_SewingMachineFlowExample}
\end{center}
\end{figure}

\begin{figure}[!t]
\begin{center}
\includegraphics[width=0.8\textwidth]{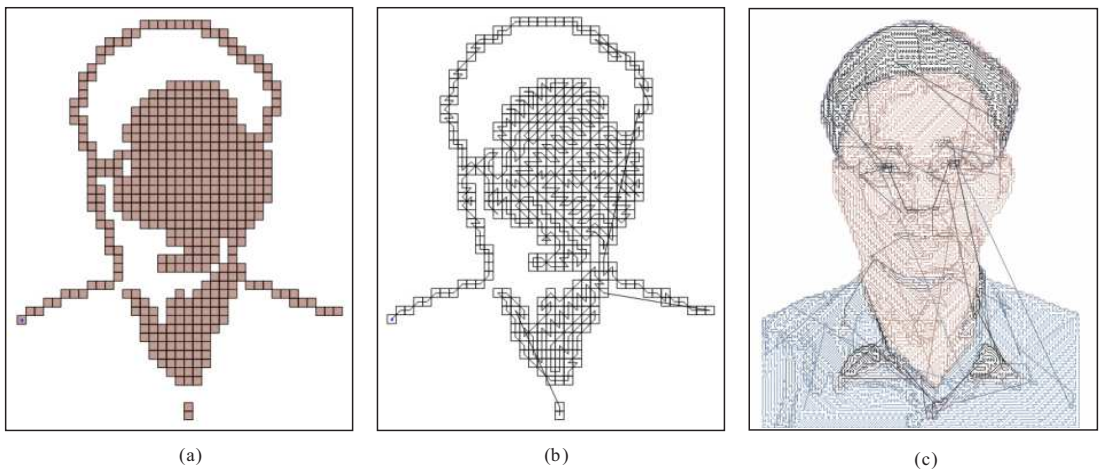}
\caption{(a) A set of lattices for one region of color, (b) a possible sewing trace for the set of lattices in (a), and (c) an overview after computing sewing traces of all regions of colors.} \label{Fig_SewingMachineFlowExample1}
\end{center}
\end{figure}


In the literature, there exist two related classes of graphs, \textit{grid} and \textit{triangular grid} graphs. In a grid graph, for each vertex $v = (v_x, v_y)$ its possible adjacent vertices include $(v_x, v_y-1)$, $(v_x-1, v_y)$, $(v_x+1, v_y)$, and $(v_x, v_y+1)$. And for each vertex $v = (v_x, v_y)$ in a triangular grid graph, its possible adjacent vertices include $(v_x, v_y-1)$, $(v_x-1, v_y)$, $(v_x+1, v_y)$, $(v_x, v_y+1)$, $(v_x-1, v_y-1)$, and $(v_x+1, v_y+1)$. Thus, supergrid graphs contain grid and triangular grid graphs as subgraphs. However, grid and triangular grid graphs are not subclasses of supergrid graphs, and the converse is also true: these classes of graphs have common elements (vertices) but in general they are distinct since the edge sets of these graphs are different. Obviously, all grid graphs are bipartite \cite{Itai82} but triangular grid graphs and supergrid graphs are not always bipartite.

A \textit{Hamiltonian path} (resp. \textit{cycle}) of a graph is a simple path (resp. \textit{cycle}) in which each vertex of the graph appears exactly once. The \textit{Hamiltonian path \emph{(resp.,} cycle\emph{)} problem} involves deciding whether or not a graph contains a Hamiltonian path (resp., cycle). A graph is said to be \textit{Hamiltonian} if it contains a Hamiltonian cycle. A graph $G$ is said to be \textit{Hamiltonian connected} if for each pair of distinct vertices $u$ and $v$ of $G$, there exists a Hamiltonian path between $u$ and $v$ in $G$. If $(u, v)$ is an edge of a Hamiltonian connected graph, then a Hamiltonian cycle containing $(u, v)$ does exist. Thus, a Hamiltonian connected graph contains many Hamiltonian cycles, and, hence, the sufficient conditions of Hamiltonian connectivity are stronger than those of Hamiltonicity. The longest $(s, t)$-path problem is to find a longest path from vertex $s$ to vertex $t$ of a graph, where $s$ and $t$ are any two given vertices and the longest path is a simple path with the maximum number of vertices. It is well known that the Hamiltonian and longest $(s, t)$-path problems are NP-complete for general graphs \cite{GareyJ79, Johnson85}. The same holds true for bipartite graphs \cite{Krishnamoorthy76}, split graphs \cite{Golumbic80}, circle graphs \cite{Damaschke89}, undirected path graphs \cite{BertossiB86}, grid graphs \cite{Itai82}, triangular grid graphs \cite{Gordon08}, supergrid graphs \cite{Hung15}, and so on. In the literature, there are many studies for the Hamiltonian connectivity of interconnection networks, see \cite{Chen00, Chen04, Huang00, Huang02, Hung12, Li09, Liu11, Lo01}.

Previous related works are summarized as follows. Recently, Hamiltonian path (cycle) and Hamiltonian connected problems in grid, triangular grid, and supergrid graphs have received much attention. Itai \textit{et al.} \cite{Itai82} showed that the Hamiltonian path and cycle problems for grid graphs are NP-complete. They also gave the necessary and sufficient conditions for a rectangular grid graph to be Hamiltonian connected. Thus, rectangular grid graphs are not always Hamiltonian connected. Zamfirescu \textit{et al.} \cite{Zamfirescu92} gave the sufficient conditions for a grid graph having a Hamiltonian cycle, and proved that all grid graphs of positive width have Hamiltonian line graphs. Later, Chen \textit{et al.} \cite{Chen02} improved the Hamiltonian path algorithm of \cite{Itai82} on rectangular grid graphs and presented a parallel algorithm for the Hamiltonian path problem with two given end vertices in rectangular grid graph. Also Lenhart and Umans \cite{Lenhart97} showed the Hamiltonian cycle problem on solid grid graphs, which are grid graphs without holes, is solvable in polynomial time. Recently, Keshavarz-Kohjerdi \textit{et al.} \cite{Keshavarz12b, Keshavarz13} presented linear-time and parallel algorithms to compute the longest path between two given vertices in rectangular grid graphs. Reay and Zamfirescu \cite{Reay00} proved that all 2-connected, linear-convex triangular grid graphs contain Hamiltonian cycles except one special case. The Hamiltonian cycle and path problems on triangular grid graphs were known to be NP-complete \cite{Gordon08}. In addition, the Hamiltonian cycle problem on hexagonal grid graphs has been shown to be NP-complete \cite{Islam07}. Alphabet grid graphs first appeared in \cite{Salman05}, in which Salman determined the classes of alphabet grid graphs containing Hamiltonian cycles. Keshavarz-Kohjerdi and Bagheri \cite{Keshavarz12a} gave the necessary and sufficient conditions for the existence of Hamiltonian paths in alphabet grid graphs, and presented a linear-time algorithm for finding Hamiltonian path with two given endpoints in these graphs. Recently, Keshavarz-Kohjerdi and Bagheri \cite{Keshavarz16} verified the Hamiltonian connectivity of $L$-shaped grid graphs. Very recently, Keshavarz-Kohjerdi and Bagheri presented a linear-time algorithm to find Hamiltonian $(s, t)$-paths in rectangular grid graphs with a rectangular hole \cite{Keshavarz17a, Keshavarz17b}, and to compute longest $(s, t)$-paths in $L$-shaped and $C$-shaped grid graphs \cite{Keshavarz18a,Keshavarz19c}. The supergrid graphs were first appeared in \cite{Hung15}, in which we proved that the Hamiltonian cycle and path problems on supergrid graphs are NP-complete, and every rectangular supergrid graph is Hamiltonian. Since the Hamiltonian cycle and path problems are NP-complete for supergrid graphs \cite{Hung15}, an important line of investigation is to discover the complexities of the Hamiltonian related problems when the input is restricted to be in special subclasses of supergrid graphs. In \cite{Hung16}, we showed that the Hamiltonian cycle problem for linear-convex supergrid graphs is linear solvable. In \cite{Hung17a}, we proved that rectangular supergrid graphs are always Hamiltonian connected except one trivial forbidden condition. Some shaped supergrid graphs have been verified to be Hamiltonian and Hamiltonian connected \cite{Hung17b}. Recently, we showed the Hamiltonian connectivity of alphabet supergrid graphs \cite{Hung19a}. Very recently, we showed that the Hamiltonian and longest $(s, t)$-paths of $L$- and $C$-shaped supergrid graphs can be computed in linear time \cite{Hung18, Keshavarz19a, Keshavarz19b}. In this paper, we first show that $O$-shaped supergrid graphs, which are rectangular supergrid graphs with rectangular holes, are always Hamiltonian and Hamiltonian connected except few conditions. We then give a linear-time algorithm to solve the longest $(s, t)$-path problem on $O$-shaped supergrid graphs. This study can be regarded as the first attempt for solving the Hamiltonian and longest $(s, t)$-path problems on hollow supergrid graphs.

The Hamiltonian connectivity of $O$-shaped supergrid graphs can be also applied to compute the minimum trace of 3D printers as follows. Consider a 3D printer with a hollow object ($O$-type object) being printed. The software produces a series of thin layers, designs a path for each layer, combines these paths of produced layers, and transmits the above paths to 3D printer. Because 3D printing is performed layer by layer, each layer can be considered as an $O$-shaped supergrid graph. Suppose that there are $k$ layers under the above 3D printing. If the Hamiltonian connectivity of $O$-shaped supergrid graphs holds true, then we can find a Hamiltonian $(s_i, t_i)$-path of an  $O$-shaped supergrid graph $O_i$, where $O_i$, $1\leqslant i\leqslant k$, represents a layer under 3D printing. Thus, we can design an optimal trace for the above 3D printing, where $t_i$ is adjacent to $s_{i+1}$ for $1\leqslant i\leqslant k-1$. In this application, we restrict the 3d printer nozzle to be located at integer coordinates. 

The paper is organized as follows. In Section \ref{Sec_Preliminaries}, some notations, observations, and previous established results are introduced. We also verify the Hamiltonicity of $O$-shaped supergrid graphs in this section. In Section \ref{Sec_forbidden-conditions}, we discover some conditions such that $O$-shaped supergrid graphs contain no Hamiltonian $(s, t)$-path. Section \ref{Sec_O-shaped-supergrid} shows that $O$-shaped supergrid graphs are Hamiltonian connected except the forbidden conditions in Section \ref{Sec_forbidden-conditions}. In Section \ref{Sec_Algorithm}, we present a linear-time algorithm to compute the longest $(s, t)$-paths of $O$-shaped supergrid graphs. Finally, we make some concluding remarks in Section \ref{Sec_Conclusion}.

\section{Terminologies and Background Results}\label{Sec_Preliminaries}
In this section, we will introduce some terminologies and symbols. Some observations and previously established results for the Hamiltonicity and Hamiltonian connectivity of rectangular supergrid graphs are also presented. For graph-theoretic terminology not defined in this paper, the reader is referred to \cite{Bondy76}.

Let $G = (V, E)$ be a supergrid graph with vertex set $V(G)$ and edge set $E(G)$. Let $S$ be a subset of vertices in $G$, and let $u$ and $v$ be two vertices in $G$. We write $G[S]$ for the subgraph of $G$ \textit{induced} by $S$, $G-S$ for the subgraph $G[V-S]$, i.e., the subgraph induced by $V-S$. In general, we write $G-v$ instead of $G-\{v\}$.
We say that $u$ is \textit{adjacent} to $v$, and $u$ and $v$ are \textit{incident} to edge $(u, v)$, if $(u, v)\in E(G)$. The notation $u\thicksim v$ (resp., $u \nsim v$) means that vertices $u$ and $v$ are adjacent (resp., non-adjacent). A vertex $w$ \textit{adjoins} edge $(u, v)$ if $w\thicksim u$ and $w\thicksim v$. For two edges $e_1=(u_1, v_1)$ and $e_2=(u_2, v_2)$, if $u_1\thicksim u_2$ and $v_1\thicksim v_2$, then we say that $e_1$ and $e_2$ are \textit{parallel}, denoted by $e_1\thickapprox e_2$. For any $v\in V(G)$, a \textit{neighbor} of $v$ is any vertex that is adjacent to $v$. Let $N_G(v)$ be the set of neighbors of $v$ in $G$, and let $N_G[v]=N_G(v)\cup\{v\}$. The \textit{degree} of vertex $v$ in $G$, denoted by $deg(v)$, is the number of vertices adjacent to $v$. A path $P$ of length $|P|$ in $G$, denoted by $v_1\rightarrow v_2\rightarrow \cdots \rightarrow v_{|P|-1} \rightarrow v_{|P|}$, is a sequence $(v_1, v_2, \cdots, v_{|P|-1}, v_{|P|})$ of vertices such that $(v_i,v_{i+1})\in E(G)$ for $1 \leqslant i < |P|$, and all vertices except $v_1, v_{|P|}$ in it are distinct. The first and last vertices visited by $P$ are denoted by $start(P)$ and $end(P)$, respectively. We will use $v_i \in P$ to denote ``$P$ visits vertex $v_i$" and use $(v_i, v_{i+1}) \in P$ to denote ``$P$ visits edge $(v_i, v_{i+1})$". A path from $v_1$ to $v_k$ is denoted by $(v_1, v_k)$-path. In addition, we use $P$ to refer to the set of vertices visited by path $P$ if it is understood without ambiguity. A cycle is a path $C$ with $|V(C)| \geqslant 4$ and $start(C) = end(C)$. Two paths (or cycles) $P_1$ and $P_2$ of graph $G$ are called \textit{vertex-disjoint} if $V(P_1)\cap V(P_2) = \emptyset$. If $end(P_1)\thicksim start(P_2)$, then two vertex-disjoint paths $P_1$ and $P_2$ can be concatenated into a path, denoted by $P_1 \Rightarrow P_2$.

The \emph{two-dimensional supergrid graph} $S^\infty$ is the infinite graph whose vertex set consists of all points of the plane with integer coordinates and in which two vertices are adjacent if the difference of their $x$ or $y$ coordinates is not larger than $1$. A \textit{supergrid graph} is a finite vertex-induced subgraph of $S^\infty$. For a vertex $v$ in a supergrid graph, it is represented as $(v_x, v_y)$, where $v_x$ and $v_y$ are the $x$ and $y$ coordinates of $v$ respectively. The possible adjacent vertices of a vertex $v=(v_x, v_y)$ in a supergrid graph hence include $(v_x, v_y-1)$, $(v_x-1, v_y)$, $(v_x+1, v_y)$, $(v_x, v_y+1)$, $(v_x-1, v_y-1)$, $(v_x+1, v_y+1)$, $(v_x+1, v_y-1)$, and $(v_x-1, v_y+1)$. The edge $(u, v)$ is said to be \textit{horizontal} (resp., \textit{vertical}) if $u_y=v_y$ (resp., $u_x=v_x$), and is called \textit{crossed} if it is neither a horizontal nor a vertical edge. Next, we define some special supergrid graphs studied in the paper as follows.

\begin{defn}
Let $R(m, n)$ be the supergrid graph whose vertex set $V(R(m, n))=\{v =(v_x, v_y) | 1\leqslant v_x\leqslant m$ and $1\leqslant v_y\leqslant n\}$. A \textit{rectangular supergrid graph} is a supergrid graph which is isomorphic to $R(m, n)$ for some $m$ and $n$, and $R(m, n)$ is called $n$-rectangle.
\end{defn}

There are four boundaries in a rectangular supergrid graph $R(m, n)$ with $m, n\geqslant 2$. The edge in the boundary of $R(m, n)$ is called \textit{boundary edge}. A path is called \textit{boundary} of $R(m, n)$ if it visits all vertices and edges of the same boundary in $R(m, n)$ and its length equals to the number of vertices in the visited boundary. Let $v=(v_x, v_y)$ be a vertex in $R(m, n)$. The vertex $v$ is called the \textit{upper-left} (resp., \textit{upper-right}, \textit{down-left}, \textit{down-right}) \textit{corner} of $R(m, n)$ if for any vertex $w=(w_x, w_y)\in R(m, n)$, $w_x\geqslant v_x$ and $w_y\geqslant v_y$ (resp., $w_x\leqslant v_x$ and $w_y\geqslant v_y$, $w_x\geqslant v_x$ and $w_y\leqslant v_y$, $w_x\leqslant v_x$ and $w_y\leqslant v_y$). Throughout this paper in the figures, $(1, 1)$ is the coordinates of the vertex in the upper-left corner, except we explicitly change this assumption.

\begin{defn}
Let $R(m, n)$ be a rectangular supergrid graph. Let $L(m, n; k, l)$ be a supergrid graph obtained from $R(m, n)$ by removing its subgraph $R(k, l)$ from the upper-right corner coordinated by $(m, 1)$. A $L$-shaped supergrid graph is isomorphic to $L(m,n; k,l)$ (see Fig. \ref{Fig_LCO-shaped}(a)).
\end{defn}

\begin{defn}
Let $R(m, n)$ be a rectangular supergrid graph. A $C$-shaped supergrid graph $C(m,n; k,l; c,d)$ is a supergrid graph obtained from a rectangular supergrid graph $R(m, n)$ by removing its subgraph $R(k, l)$ from its vertex coordinated by $(m, c+1)$ while $R(m, n)$ and $R(k, l)$ have exactly one border side in common, where $m\geqslant 2$, $n\geqslant  3$, $k, l, c, d\geqslant 1$, and $n = c+d+l$ (see Fig. \ref{Fig_LCO-shaped}(b)).
\end{defn}

\begin{defn}
Let $R(m, n)$ be a rectangular supergrid graph. An $O$-shaped supergrid graph $O(m,n; k,l; a,b,c,d)$ is a supergrid graph obtained from a rectangular supergrid graph $R(m, n)$ by removing its subgraph $R(k, l)$ from its vertex coordinated by $(m-b, c+1)$ while $R(m, n)$ and $R(k, l)$ have no border side in common, where $m, n\geqslant 3$, $k,l,a,b,c,d\geqslant 1$, $m = a+b+k$, and $n = c+d+l$ (see Fig. \ref{Fig_LCO-shaped}(c)).
\end{defn}

\begin{figure}[!t]
\begin{center}
\includegraphics[width=0.8\textwidth]{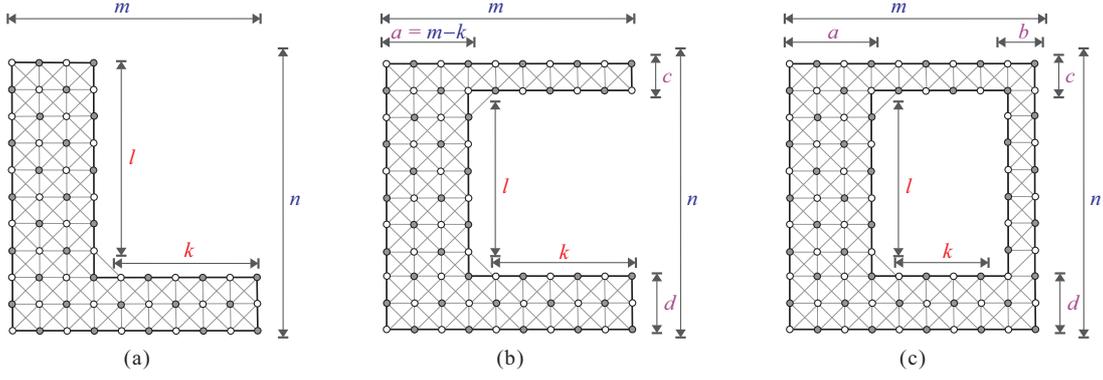}
\caption{The structure of (a) $L$-shaped supergrid graph $L(m,n; k,l)$, (b) $C$-shaped supergrid graph $C(m,n; k,l; c,d)$, and (c) $O$-shaped supergrid graph $O(m,n; k,l; a,b,c,d)$.} \label{Fig_LCO-shaped}
\end{center}
\end{figure}

In proving our results, we need to partition a supergrid graph into $k$ disjoint parts, where $k\geqslant 2$. The partition is defined as follows.

\begin{defn}
Let $G$ be a supergrid graph. A \textit{separation operation} on $G$ is a partition of $G$ into $k$ vertex-disjoint supergrid subgraphs $G_1$, $G_2$, $\cdots$, $G_k$, i.e., $V(G)=V(G_1)\cup V(G_2)\cup \cdots \cup V(G_k)$ and $V(G_i)\cap V(G_j) =\emptyset$ for $i\neq j$ and $1\leqslant i, j\leqslant k$, where $k\geqslant 2$. A separation is called \textit{vertical} if it consists of a set of horizontal edges, and is called \textit{horizontal} if it contains a set of vertical edges.
\end{defn}

Let $(G, s, t)$ denote the supergrid graph $G$ with two specified distinct vertices $s$ and $t$. Without loss of generality, we will assume that $s_x \leqslant t_x$ in the rest of the paper, except we explicitly change this assumption. We denote a Hamiltonian path between $s$ and $t$ in $G$ by $HP(G, s, t)$. We say that $HP(G, s, t)$ does exist if there is a Hamiltonian $(s, t)$-path in $G$. Next, we will introduce some previously established results.

Let $R(m, n)$ be a rectangular supergrid graph with $m\geqslant n\geqslant 2$, $\mathcal{C}$ be a cycle of $R(m, n)$, and let $H$ be a boundary of $R(m, n)$, where $H$ is a subgraph of $R(m, n)$. The restriction of $\mathcal{C}$ to $H$ is denoted by $\mathcal{C}_{| H}$. If $|\mathcal{C}_{| H}|=1$, i.e. $\mathcal{C}_{| H}$ is a boundary path on $H$, then $\mathcal{C}_{| H}$ is called \textit{flat face} on $H$. If $|\mathcal{C}_{| H}|>1$ and $\mathcal{C}_{| H}$ contains at least one boundary edge of $H$, then $\mathcal{C}_{| H}$ is called \textit{concave face} on $H$. A Hamiltonian cycle of $R(m, 3)$ is called \textit{canonical} if it contains three flat faces on two shorter boundaries and one longer boundary, and it contains one concave face on the other boundary, where the shorter boundary consists of three vertices. And, a Hamiltonian cycle of $R(m, n)$ with $n=2$ or $n\geqslant 4$ is said to be \textit{canonical} if it contains three flat faces on three boundaries, and it contains one concave face on the other boundary. The following lemma states the result in \cite{Hung15} concerning the Hamiltonicity of rectangular supergrid graphs.

\begin{lem}\label{HC-rectangular_supergrid_graphs}
(See \cite{Hung15}) Let $R(m, n)$ be a rectangular supergrid graph with $m\geqslant n\geqslant 2$. Then, the following statements hold true:\\
$(1)$ if $n=3$, then $R(m, 3)$ contains a canonical Hamiltonian cycle;\\
$(2)$ if $n=2$ or $n\geqslant 4$, then $R(m, n)$ contains four canonical Hamiltonian cycles with concave faces being on different boundaries.
\end{lem}

\begin{defn}
Assume that $G$ is a connected supergrid graph and $V_1$ is a subset of the vertex set $V(G)$. $V_1$ is a \textit{vertex cut} if $G-V_1$ is disconnected. A vertex $v\in V(G)$ is a \textit{cut vertex}, if $G-\{v\}$ is disconnected. For an example, in Fig. \ref{Fig_ForbiddenConditionF1}(a) $t$ is a cut vertex, and in Fig. \ref{Fig_ForbiddenConditionF1}(b) $\{s, t\}$ is a vertex cut.
\end{defn}

\begin{figure}[!t]
\begin{center}
\includegraphics[scale=0.9]{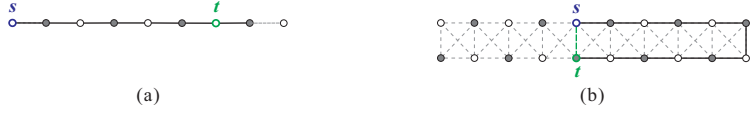}
\caption{Rectangular supergrid graphs in which there is no Hamiltonian $(s, t)$-path for (a) $R(m, 1)$, and (b) $R(m, 2)$, where solid lines indicate the longest path between $s$ and $t$.} \label{Fig_ForbiddenConditionF1}
\end{center}
\end{figure}

In \cite{Hung17a}, the authors showed that $HP(R(m,n), s, t)$ does not exist if the following condition holds:

\begin{description}
  \item[(F1)] $s$ or $t$ is a cut vertex, or $\{s, t\}$ is a vertex cut (see Fig. \ref{Fig_ForbiddenConditionF1}(a) and Fig. \ref{Fig_ForbiddenConditionF1}(b)).
\end{description}

In addition to condition (F1) (as depicted in Fig. \ref{Fig_ForbiddenCondition_L-F2F3}(a) and  \ref{Fig_ForbiddenCondition_L-F2F3}(b)), in \cite{Hung18, Keshavarz19a}, we showed that $HP(L(m,n; k,l), s, t)$ does not exist whenever one of the following conditions is satisfied.

\begin{description}
  \item[(F2)] assume that $G$ is a supergrid graph, there exists a vertex $w \in G$ such that $deg(w) = 1$, $w \neq s$, and $w \neq t$ (see Fig. \ref{Fig_ForbiddenCondition_L-F2F3}(c)).
  \item[(F3)] $m-k = 1$, $n-l = 2$, $l = 1$, $k \geq 2$, and $\{s, t\} = \{(1, 2), (2, 3)\}$ or $\{(1, 3), (2, 2)\}$ (see Fig. \ref{Fig_ForbiddenCondition_L-F2F3}(d)).
\end{description}

\begin{figure}[!t]
\begin{center}
\includegraphics[width=0.8\textwidth]{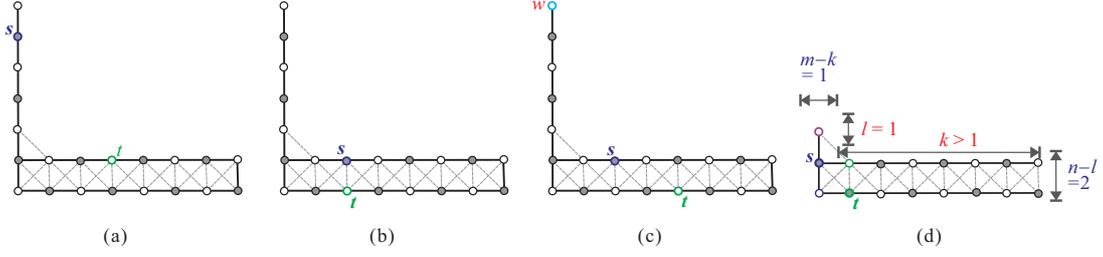}
\caption{$L$-shaped supergrid graph in which there is no Hamiltonian $(s, t)$-path for (a) $s$ is a cut vertex, (b) $\{s, t\}$ is a vertex cut, (c) there exists a vertex $w$ such that $deg(w)=1$, $w\neq s$, and $w\neq t$, and (d) $m-k=1$, $n-l=2$, $l=1$, $k\geqslant 2$, and $\{s, t\}=\{(1, 2), (2, 3)\}$.} \label{Fig_ForbiddenCondition_L-F2F3}
\end{center}
\end{figure}

In addition to conditions (F1) (as depicted in Fig. \ref{Fig_ForbiddenCondition_C-F4F6}(a)--\ref{Fig_ForbiddenCondition_C-F4F6}(b)) and (F2) (as depicted in Fig. \ref{Fig_ForbiddenCondition_C-F4F6}(c)), in \cite{Keshavarz19b}, we showed that $C(m,n; k,l; c,d)$ contains no Hamiltonian $(s, t)$-path if $(C(m,n; k,l; c,d), s, t)$ satisfies one of the following conditions.

\begin{figure}[!t]
\centering
\includegraphics[width=0.8\textwidth]{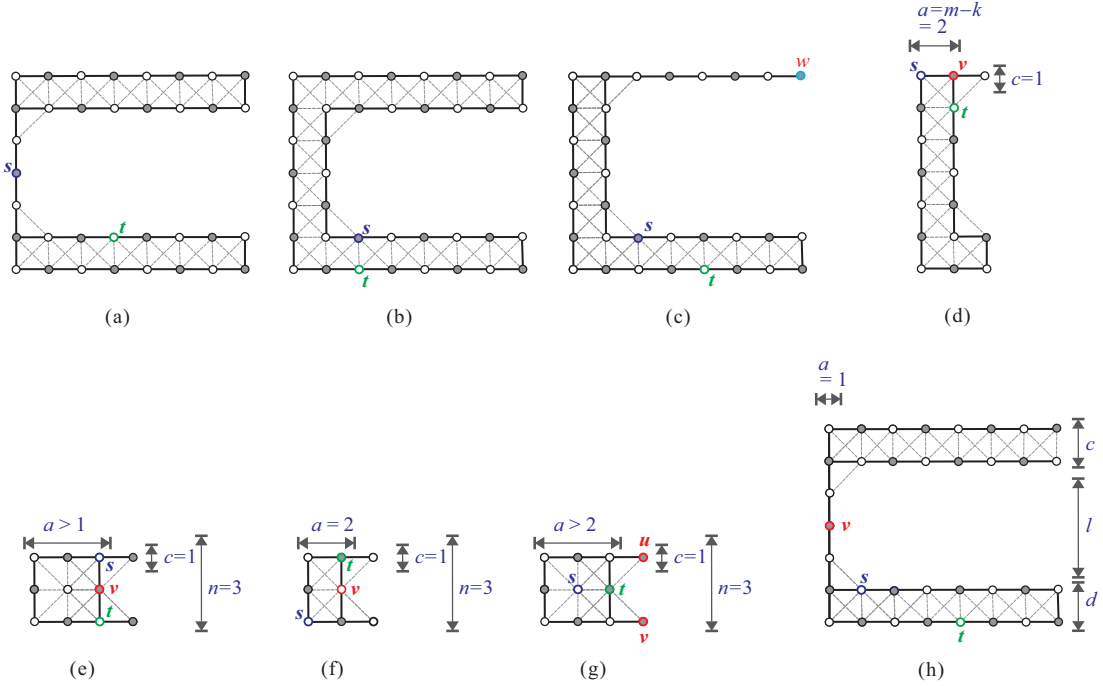}
\caption{Some $C$-shaped supergrid graphs in which there is no Hamiltonian $(s, t)$-path.}\label{Fig_ForbiddenCondition_C-F4F6}
\end{figure}

\begin{description}
  \item[(F4)] $m=3$, $a=m-k=2$, and $[(c=1$ and $\{s, t\} = \{(1, 1), (2, 2)\}$ or $\{(1, 2), (2, 1)\})$ or $(d=1$ and $\{s, t\} = \{(1, n), (2, n-1)\}$ or $\{(1, n - 1), (2, n)\})]$ (see Fig. \ref{Fig_ForbiddenCondition_C-F4F6}(d)).
  \item[(F5)] $n=3$, $k=c=d=1$, and\\
\hspace{0.65cm}(1) $a \geqslant 2$, $s_x=t_x=m-1$, and $|s_y-t_y|=2$ (see Fig. \ref{Fig_ForbiddenCondition_C-F4F6}(e)); or\\
\hspace{0.65cm}(2) $a = 2$, $s_x = 1$, $t_x = 2$, and $|s_y-t_y|=2$ (see Fig. \ref{Fig_ForbiddenCondition_C-F4F6}(f)); or\\
\hspace{0.65cm}(3) $a > 2$, $s_x < m - 1$, and $t = (m-1, 2)$ (see Fig. \ref{Fig_ForbiddenCondition_C-F4F6}(g)).
  \item[(F6)] $a=m-k=1$, and ($s_y, t_y\leqslant c$ or $s_y, t_y > c + l$) (see Fig. \ref{Fig_ForbiddenCondition_C-F4F6}(h)).
\end{description}

\begin{thm}\label{HP-Theorem-RLCshaped}(See \cite{Hung17a, Hung18, Keshavarz19a, Keshavarz19b})
Let $G$ be a supergrid graph with vertices $s$ and $t$, where $G$ is rectangular, $L$-shaped, or $C$-shaped. $(G,s,t)$ contains a Hamiltonian $(s, t)$-path, i.e., $HP(G, s, t)$ does exist if and only if $(G, s, t)$ does not satisfy conditions $\mathrm{(F1)}$--$\mathrm{(F6)}$.
\end{thm}

The Hamiltonian $(s, t)$-path $P$ of $R(m, n)$ constructed in \cite{Hung17a} satisfies that $P$ contains at least one boundary edge of each boundary, and is called \textit{canonical}.

\begin{lem}\label{HamiltonianConnected-Rectangular}
(See \cite{Hung17a}) Let $R(m, n)$ be a rectangular supergrid graph with $m, n \geqslant 1$, and let $s$ and $t$ be its two distinct vertices. If $(R(m, n), s, t)$ does not satisfy condition $\mathrm{(F1)}$, then there exists a canonical Hamiltonian $(s, t)$-path of $R(m, n)$, i.e., $HP(R(m, n), s, t)$ does exist.
\end{lem}

Consider that $(R(m, n), s, t)$ does not satisfy condition (F1). Let $w=(1, 1)$, $z=(2, 1)$, and $f=(3, 1)$ be three vertices of $R(m, n)$ with $m\geqslant 3$ and $n\geqslant 2$. In \cite{Keshavarz19a}, we proved that there exists a Hamiltonian $(s, t)$-path $Q$ of $R(m, n)$ such that $(z, f)\in Q$ if the following condition (F7) is satisfied; and $(w, z)\in Q$ otherwise.

\begin{description}
  \item[(F7)] $n=2$ and $\{s, t\}\in \{\{w, z\}, \{(1, 1), (2, 2)\}, \{(2, 1), (1, 2)\}\}$, or $n\geqslant 3$ and $\{s, t\}=\{w, z\}$.
\end{description}

The above result is presented as follows.

\begin{lem}\label{HamiltonianConnected-Rectangular-wz_rectangle}
(See \cite{Keshavarz19a}) Let $R(m, n)$ be a rectangular supergrid graph with $m\geqslant 3$ and $n\geqslant 2$, $s$ and $t$ be its two distinct vertices, and let $w=(1, 1)$, $z=(2, 1)$, and $f=(3, 1)$. If $(R(m, n), s, t)$ does not satisfy condition $\mathrm{(F1)}$, then there exists a canonical Hamiltonian $(s, t)$-path $Q$ of $R(m, n)$ such that $(z, f)\in Q$ if $(R(m, n), s, t)$ does satisfy condition $\mathrm{(F7)}$; and $(w, z)\in Q$ otherwise.
\end{lem}

For a 3-rectangle $R(m, 3)$, we obtain the following lemma in \cite{Keshavarz19b}.


\begin{lem}\label{HP-3rectangle-boundary_path}
(See \cite{Keshavarz19b}) Let $R(m, 3)$ be a $3$-rectangle with $m\geqslant 3$, and let $s$ and $t$ be its two distinct vertices. Let $z_1=(m, 1)$, $z_2=(m, 2)$, and $z_3=(m, 3)$ be three vertices of $R(m, 3)$, and let edges $e_{12}=(z_1, z_2)$, $e_{23}=(z_2, z_3)$. If $\{s, t\}\cap\{z_1, z_2, z_3\}=\emptyset$, then there exists a Hamiltonian $(s, t)$-path of $R(m, 3)$ containing $e_{12}$ and $e_{23}$.
\end{lem}

In \cite{Hung18} and \cite{Keshavarz19b}, we verified the Hamiltonicity of $L$-shaped and $C$-shaped supergrid graphs as follows.

\begin{thm}\label{HC-LCshaped}
(See \cite{Hung18, Keshavarz19a, Keshavarz19b}) Let $L(m,n; k,l)$ (resp. $C(m,n; k,l; c,d)$) be a $L$-shaped (resp. $C$-shaped) supergrid graph. Then, $L(m,n; k, l)$ (resp. $C(m,n; k,l; c,d)$) contains a Hamiltonian cycle if it does not satisfy condition $\mathrm{(F8)}$ (resp. $\mathrm{(F9)}$), where condition $\mathrm{(F8)}$ (resp. $\mathrm{(F9)}$) is defined as follows:
\end{thm}

\begin{description}
  \item[(F8)] there exists a vertex $w$ in $L(m,n; k,l)$ such that $deg(w)=1$.
  \item[(F9)] $a(=m-k)=1$ or there exists a vertex $w \in V(C(m,n; k,l; c,d))$ such that $deg(w)=1$.
\end{description}

We then give some observations on the relations among cycle, path, and vertex. These propositions will be used in proving our results and are given in \cite{Hung15, Hung16, Hung17a}.

\begin{pro}\label{Pro_Obs}
(See \cite{Hung15, Hung16, Hung17a}) Let $C_1$ and $C_2$ be two vertex-disjoint cycles of a graph $G$, let $C_1$ and $P_1$ be a cycle and a path, respectively, of $G$ with $V(C_1)\cap V(P_1)=\emptyset$, and let $x$ be a vertex in $G-V(C_1)$ or $G-V(P_1)$. Then, the following statements hold true:\\
$(1)$ If there exist two edges $e_1\in C_1$ and $e_2\in C_2$ such that $e_1 \thickapprox e_2$, then $C_1$ and $C_2$ can be combined into a cycle of $G$ (see Fig. \emph{\ref{Fig_Obs}(a)}).\\
$(2)$ If there exist two edges $e_1\in C_1$ and $e_2\in P_1$ such that $e_1 \thickapprox e_2$, then $C_1$ and $P_1$ can be combined into a path of $G$ (see Fig. \emph{\ref{Fig_Obs}(b)}). \\
$(3)$ If vertex $x$ adjoins one edge $(u_1, v_1)$ of $C_1$ (resp., $P_1$), then $C_1$ (resp., $P_1$) and $x$ can be combined into a cycle (resp., path) of $G$ (see Fig. \emph{\ref{Fig_Obs}(c)}).\\
$(4)$ If there exists one edge $(u_1, v_1)\in C_1$ such that $u_1\thicksim start(P_1)$ and $v_1\thicksim end(P_1)$, then $C_1$ and $P_1$ can be combined into a cycle $C$ of $G$ (see Fig. \emph{\ref{Fig_Obs}(d)}).
\end{pro}

\begin{figure}[!t]
\begin{center}
\includegraphics[width=0.7\textwidth]{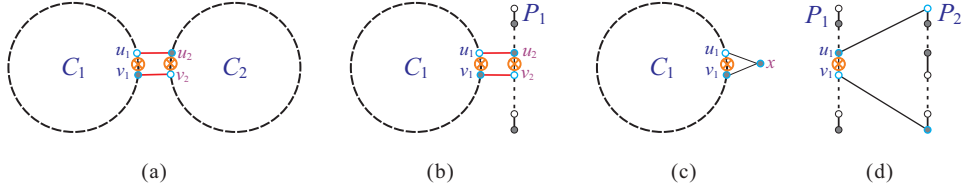}
\caption{A schematic diagram for (a) Statement (1), (b) Statement (2), (c) Statement (3), and (d) Statement (4) of Proposition \ref{Pro_Obs}, where bold dashed lines indicate the cycles (paths) and $\otimes$ represents the destruction of an edge while constructing a cycle or path.} \label{Fig_Obs}
\end{center}
\end{figure}

For the longest $(s, t)$-path problem on $R(m, n)$, $L(m,n; k,l)$, and $C(m,n; k,l; c,d)$, we showed in \cite{Hung17a, Keshavarz19a, Keshavarz19b} that it can be solved in linear time.

\begin{thm}\cite{Hung17a, Keshavarz19a, Keshavarz19b}\label{RLC-shaped_LongestPath}
Given a rectangular supergrid graph $R(m, n)$ with $mn\geqslant 2$, $L$-shaped supergrid graph $L(m,n; k,l)$, or $C$-shaped supergrid graph $C(m,n; k,l; c,d)$, and two distinct vertices $s$ and $t$ in $R(m, n)$, $L(m,n; k,l)$ or $C(m,n; k,l; c,d)$, a longest $(s, t)$-path can be computed in $O(mn)$-linear time.
\end{thm}

In this paper, we will study $O$-shaped supergrid graph $O(m,n; k,l; a,b,c,d)$ whose structure is depicted in Fig. \ref{Fig_LCO-shaped}(c). By symmetry, we will only consider the following three cases, the isomorphic cases are omitted.

\begin{verse}
\textbf{(1)} $a=b=c=d=1$; or\\
\textbf{(2)} $a\geqslant 2$ and $c=1$; or\\
\textbf{(3)} $a,b,c,d\geqslant 2$.\\
\end{verse}

The above case (2) contains the following four subcases: (2.1) $b=d=1$, (2.2) $b=1$ and $d\geqslant 2$, (2.3) $b\geqslant 2$ and $d=1$, and (2.4) $b, d\geqslant 2$. These four subcases are depicted in Fig. \ref{Fig_CaseII}. Depending on the positions of $s$ and $t$ in $O(m,n; k,l; a,b,c,d)$, we can consider the following cases for the above three cases : $(s_x, t_x\leqslant a)$, $(s_x\leqslant a$ and $t_x\geqslant a+1)$, or $(s_x, t_x\geqslant a+1)$.

\begin{figure}[!t]
\begin{center}
\includegraphics[width=0.85\textwidth]{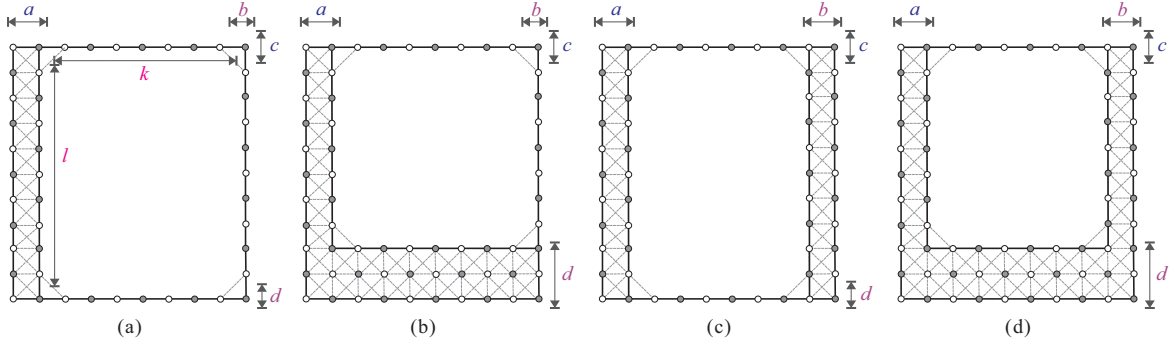}
\caption{The subcases of $a\geqslant 2$ and $c=1$, where (a) $b=d=1$, (b) $b=1$ and $d\geqslant 2$, (c) $b\geqslant 2$ and $d=1$, and (d) $b, d\geqslant 2$.} \label{Fig_CaseII}
\end{center}
\end{figure}

We first verify the Hamiltonicity of $O$-shaped supergrid graphs as the following theorem.

\begin{thm}\label{HC-Oshaped}
Let $O(m,n; k,l; a,b,c,d)$ be an $O$-shaped supergrid graph. Then, $O(m,n; k,l; a,b,c,d)$ always contains a Hamiltonian cycle.
\end{thm}
\begin{proof}
We first make a vertical separation on $O(m,n; k,l; a,b,c,d)$ to obtain two disjoint supergrid subgraphs $R_1 = R(a, n)$ and $R_2 = C(m-a,n; k,l; c,d)$, as shown in Fig. \ref{Fig_HC}(a). Let $p=(a, 1)$ and $q=(a, n)$ be two vertices of $R_1$, and let $r_1=(a+1, 1)$ and $r_2=(a+1, n)$ be two vertices of $R_2$, as depicted in Fig. \ref{Fig_HC}(b). Then, $p\thicksim r_1$ and $q\thicksim r_2$. Since $p$ and $q$ are corners of $R_1$, $(R_1, p, q)$ does not satisfy condition (F1). By Lemma \ref{HamiltonianConnected-Rectangular}, $R_1$ contains a Hamiltonian $(p, q)$-path $P_1$. By inspecting conditions (F1)--(F2) and (F4)--(F6), $(R_2, r_2, r_1)$ does not satisfy these conditions. By Theorem \ref{HP-Theorem-RLCshaped}, $(R_2, r_2, r_1)$ contains a Hamiltonian $(r_2, r_1)$-path $P_2$. Then, $P=P_1\Rightarrow P_2$ forms a Hamiltonian $(p, r_1)$-path of $O(m,n; k,l; a,b,c,d)$. Since $start(P)\thicksim end(P)$, $P$ is a Hamiltonian cycle of $O(m,n; k,l; a,b,c,d)$. The constructed Hamiltonian cycle is depicted in Fig. \ref{Fig_HC}(b).

\begin{figure}[!t]
\begin{center}
\includegraphics[scale=0.85]{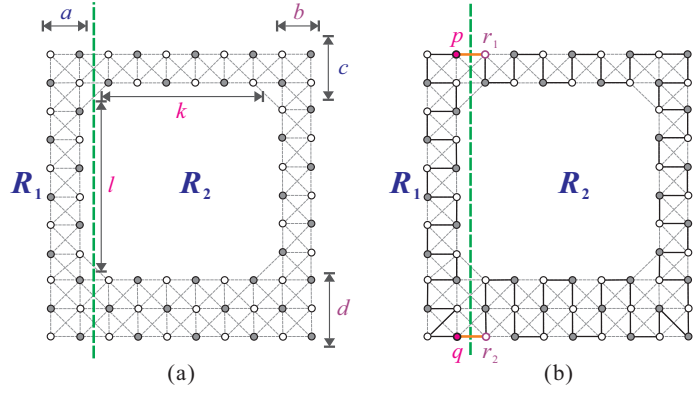}
\caption{(a) A vertical separation on $O(m,n; k,l; a,b,c,d)$ to obtain $R_1$ and $R_2$, and (b) the constructed Hamiltonian cycle of $O(m,n; k,l; a,b,c,d)$, where solid bold lines indicate the constructed Hamiltonian cycle.} \label{Fig_HC}
\end{center}
\end{figure}
\end{proof}

We can see from the above construction that if $a = 1$ then the constructed Hamiltonian cycle $HC$ of $O(m,n; k,l;$ $a,b,c,d)$ satisfies that $HC_{| R_1}$ is a flat face.

\section{The Forbidden Conditions for the Hamiltonian Connectivity of $O$-shaped Supergrid Graphs}\label{Sec_forbidden-conditions}
In this section, we will discover all cases for that $O$-shaped supergrid graphs contain no Hamiltonian $(s, t)$-path. By the structure of $O$-shaped supergrid graphs, there exists no cut vertex in them. However, there exist vertex cuts in an $O$-shaped supergrid graph. Thus, we have the forbidden condition (F1) for that $HP(O(m,n; k,l; a,b,c,d), s, t)$ does not exist (see Fig. \ref{Fig_ForbiddenCondition-oshaped1}(a) and \ref{Fig_ForbiddenCondition-oshaped1}(b)).\\


In the following, we would like to probe the other forbidden conditions for that $HP(O(m,n; k,l; a,b,c,d), s, t)$ does not exist. We consider the sizes of parameters $a$, $b$, $c$, $d$ and list the forbidden conditions as follows.

\begin{figure}[!t]
\begin{center}
\includegraphics[scale=0.85]{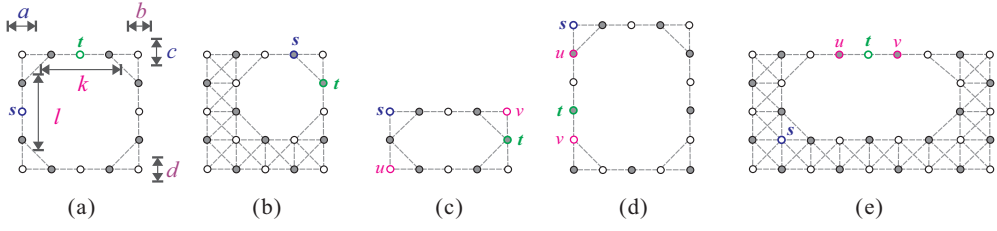}
\caption{$O$-shaped supergrid graphs in which there is no Hamiltonian $(s, t)$-path for (a)--(b) $\{s, t\}$ is a vertex cut, (c)--(d) $a=b=c=d=1$, and (e) $c=1$, $a,b,d\geqslant 2$ and $k\geqslant 5$.} \label{Fig_ForbiddenCondition-oshaped1}
\end{center}
\end{figure}

\begin{description}
  \item[(F10)] $a=b=c=d=1$, $s\not\thicksim t$, $\{s,t\}$ is not a vertex cut, and one of the following cases occurs:\\
  \hspace{0.65cm}(1) $s\in \{(1, 1), (1, n)\}$, $s_x\neq t_x$, and $s_y\neq t_y$ (see Fig. \ref{Fig_ForbiddenCondition-oshaped1}(c));\\
  \hspace{0.65cm}(2) $l\geqslant 3$, $s \ \mathrm{or}\ t\in \{(1, 1),(1, n)\}$, $s_x=t_x$, and $|s_y-t_y|>2$ (see Fig. \ref{Fig_ForbiddenCondition-oshaped1}(d)).

  \item[(F11)] $c=1$, $a, b, d\geqslant 2$, $k\geqslant 5$, $s\not\thicksim t$, $\{s,t\}$ is not a vertex cut, and $a+3\leqslant s_x(\mathrm{or}\ t_x)\leqslant a+k-2$  (see Fig. \ref{Fig_ForbiddenCondition-oshaped1}(e)).

  \item[(F12)] $c=d=1$, $a, b\geqslant 2$, $k\geqslant 3$, and one of the following cases occurs:\\
  \hspace{0.65cm}(1) $s_x\leqslant a$ and $t_x\geqslant a+3$ (see Fig. \ref{Fig_ForbiddenCondition-oshaped2}(a));\\
  \hspace{0.65cm}(2) $a+1\leqslant s_x\leqslant a+k-2$ and $t_x\geqslant a+k+1$ (see Fig. \ref{Fig_ForbiddenCondition-oshaped2}(b))).

  \item[(F13)] $b=c=1$, $a\geqslant 2$,  $s\not\thicksim t$, and one of the following cases occurs:\\
  \hspace{0.65cm}(1) $s_x,t_x\geqslant a+1$, $t=(m,1)$, $[(d=1)$ or $(d\geqslant 2$ and $s_y,t_y\leqslant c+l)]$, and

  \hspace{0.4cm}(1.1) $k\geqslant 3$, $s_y=1$, and $s_x\leqslant m-3$ (see Fig. \ref{Fig_ForbiddenCondition-oshaped2}(c) and Fig. \ref{Fig_ForbiddenCondition-oshaped2}(d)); or

  \hspace{0.4cm}(1.2) $s_y=n$ and $s_x\leqslant m-1$ (see Fig. \ref{Fig_ForbiddenCondition-oshaped3}(a)); or

  \hspace{0.4cm}(1.3) $l\geqslant 3$, $s_x=t_x$, and $s_y\geqslant 4$ (see Fig. \ref{Fig_ForbiddenCondition-oshaped3}(b) and Fig. \ref{Fig_ForbiddenCondition-oshaped3}(c));\\
  \hspace{0.65cm}(2) $d\geqslant 2$, $t_x\geqslant a+1$, $t_y\leqslant c+l$, $[(s_x\leqslant a)$ or $(s_x\geqslant a+1$ and $s_y\geqslant c+l+1)]$, and

  \hspace{0.4cm}(2.1) $k\geqslant 3$, $t_y=1$, and $a+3\leqslant t_x\leqslant m-1$ (see Fig. \ref{Fig_ForbiddenCondition-oshaped3}(d)); or

  \hspace{0.4cm}(2.2) $l\geqslant 3$, $t_x=m$, and $2\leqslant t_y\leqslant l-1$ (see Fig. \ref{Fig_ForbiddenCondition-oshaped3}(e)); or

  \hspace{0.4cm}(2.3) $k,l\geqslant 3$ and $t=(m,1)$ (see Fig. \ref{Fig_ForbiddenCondition-oshaped3}(f));\\
  \hspace{0.65cm}(3) $ d=1$, $s_x\leqslant a$, $t_x\geqslant a+1$, and $[(2\leqslant t_y\leqslant c+l)$ or $(k\geqslant 3$ and $t_x\geqslant a+3)]$ (see Fig. \ref{Fig_ForbiddenCondition-oshaped3}(g) and Fig. \ref{Fig_ForbiddenCondition-oshaped3}(h)).
\end{description}

\begin{figure}[!t]
\centering
\includegraphics[scale=0.85]{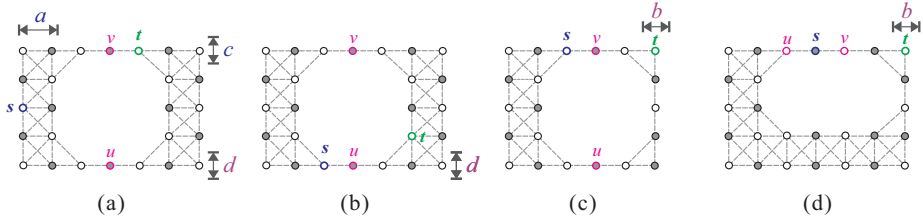}
\caption{$O$-shaped supergrid graphs with $a\geqslant 2$ and $c=1$ in which there is no Hamiltonian $(s, t)$-path for (a) $d=1$, $b\geqslant 2$, $s_x\leqslant a$, and $t_x\geqslant a+3$, (b) $d=1$, $b\geqslant 2$, $a+1\leqslant s_x\leqslant a+k-2$, and $t_x\geqslant a+k+1$, and (c)--(d) $b=1$, $s_x, t_x\geqslant a+1$, $t=(m, 1)$, $k\geqslant 3$, and $s_x\leqslant m-3$.} \label{Fig_ForbiddenCondition-oshaped2}
\end{figure}

\begin{figure}[!t]
\centering
\includegraphics[scale=0.85]{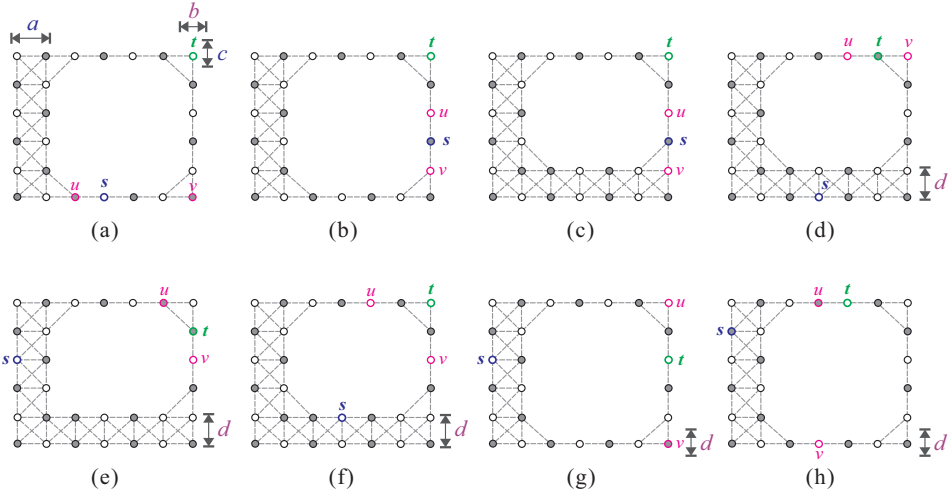}
\caption{$O$-shaped supergrid graphs with $a\geqslant 2$ and $c=b=1$ in which there is no Hamiltonian $(s, t)$-path.} \label{Fig_ForbiddenCondition-oshaped3}
\end{figure}

The following lemma shows the necessary condition for that $HP(O(m,n; k,l; a, b,c,d), s, t)$ does exist.

\begin{lem}\label{Necessary-condition-Oshaped}
If $HP(O(m,n; k,l; a,b,c,d), s, t)$ exists, then $(O(m,n; k,l; a,b,c,d)$ does not satisfy conditions $\mathrm{(F1)}$ and $\mathrm{(F10)}$--$\mathrm{(F13)}$.
\end{lem}
\begin{proof}
Assume that $(O(m,n; k,l; a,b,c,d), s, t)$ satisfies one of the conditions (F1) and (F10)--(F13), then we show that $HP(O(m,n; k,l; a,b,c,d), s, t)$ does not exist. For condition (F1), it the lemma clearly holds true (see Fig. \ref{Fig_ForbiddenCondition-oshaped1}(a) and \ref{Fig_ForbiddenCondition-oshaped1}(b)). For conditions (F10)--(F13), consider Figs. \ref{Fig_ForbiddenCondition-oshaped1}(c)--(e), Fig. \ref{Fig_ForbiddenCondition-oshaped2}, and Fig. \ref{Fig_ForbiddenCondition-oshaped3}. Let $v$ and $u$ be two vertices depicted in these figures. It is easy to see that there is no Hamiltonian $(s, t)$-path in $O(m,n; k,l; a,b,c,d)$ containing both of vertices $u$ and $v$.
 \end{proof}

We have considered any case to discover the forbidden conditions for that $HP(O(m,n; k,l; a,b,c,d), s, t)$ does exist. In the next section, we will verify that $O(m,n; k,l; a,b,c,d)$ contains a Hamiltonian $(s, t)$-path if $(O(m,n; k,l; a,b,c,d),$ $s, t)$ does not satisfy conditions (F1) and (F10)--(F13).

\section{The Hamiltonian Connectivity of $O$-shaped Supergrid Graphs}\label{Sec_O-shaped-supergrid}
In this section, we will show that $O(m,n; k,l; a,b,c,d)$ always contains a Hamiltonian $(s, t)$-path when $(O(m,n; k,l;$ $a,b, c,d), s, t)$ does not satisfy conditions (F1) and (F10)--(F13). Note that in the following lemmas, we consider the cases that $(s_x,t_x\leqslant a)$, $(s_x,t_x\geqslant a+1)$, and $(s_x\leqslant a$ and $t_x\geqslant a+1)$.

\begin{lem}\label{HP-Oshaped1}
Let $O(m,n; k,l; a,b,c, d)$ be an $O$-shaped supergrid graph, and let $s$ and $t$ be its two distinct vertices such that $s_x, t_x\leqslant a$ and $(O(m,n; k,l; a,b,c, d), s, t)$ does not satisfy conditions $\mathrm{(F1)}$ and $\mathrm{(F10)}$. Then, $O(m,n; k,l; a,b,c,d)$ contains a Hamiltonian $(s, t)$-path, i.e., $HP(O(m,n; k,l; a,b,c,d), s, t)$ does exist.
\end{lem}
\begin{proof}
By symmetry, we can only consider the cases of $a=b=c=d=1$, $a\geqslant 2$ and $c=1$, and $a,b,c,d\geqslant 2$, as illustrated in Section \ref{Sec_Preliminaries}, the isomorphic cases are omitted. We then consider the following three cases:

Case 1: $a=b=c=d=1$. We first make a vertical separation on $O(m,n; k,l; a,b,c,d)$ to obtain two disjoint supergrid subgraphs $R_1=R(m_1,n)$ and $R_2 = C(m-m_1,n; k,l; c,d)$, where $m_1=a$ (see Fig. \ref{Fig_HP1}(a)). Then, $(s \thicksim t)$ or $(s\not\thicksim t$ and $((l\leqslant 2)$ or $(l\geqslant 3$ and $[(s_y,t_y\leqslant 3)$ or $(s_y,t_y\geqslant n-2)])))$. If $s\not\thicksim t$, $l\geqslant 3$, and $[(s, t\notin \{(1,1),(1,n)\})$ or $(s$ or $t\in \{(1,1),(1,n)\}$ and $|s_y-t_y|>2)]$, then $(O(m, n; k, l; a,b,c, d), s, t)$ satisfies condition (F1) or (F10). Without loss of generality, assume that $t_y< s_y$ and for the case $s\not\thicksim t$ assume that $t=(1,1)$. We make a horizontal separation on $R_1$ to obtain two disjoint supergrid subgraphs $R_{11}$ and $R_{12}$ such that $R_{11}=R(m_1,n_1)$ and $R_{12}=R(m_1,n-n_1)$, where $n_1=t_y$ if $s \thicksim t$; otherwise $n_1=t_y+1$ (see Fig. \ref{Fig_HP1}(b) and Fig. \ref{Fig_HP1}(c)). Let $q\in V(R_{11})$, $p\in V(R_{12})$, and $w,z \in V(R_2)$ such that such that $p\thicksim w$, $q\thicksim z$, $z=(m_1+1,1)$, $w=(m_1+1,n)$, $p=(m_1, n)$, and $q=(m_1, 1)$ if $s\thicksim t$; otherwise $q=(m_1, 2)$. Consider $(R_2, w, z)$. Since $z_x =w_x=m_1+1=a+1$, $z_y=1$, and $w_y=n$, clearly $(R_2, w, z)$ does not satisfy conditions (F1), (F2), and (F4)--(F6). Now, consider $(R_{11}, q, t)$ and $(R_{12}, s, p)$. Clearly if $(R_{11}, q, t)$ or $(R_{12}, s, p)$ satisfies condition (F1), then $(O(m,n; k,l; a,b,c, d), s, t)$ satisfies condition (F10), a contradiction. Thus, $(R_{11}, q, t)$ and $(R_{12}, s, p)$ do not satisfy condition (F1). Since $(R_{12}, s, p)$, $(R_{11}, q, t)$, and $(R_2, w, z)$ do not satisfy conditions (F1)--(F2) and (F4)--(F6), By Theorem \ref{HP-Theorem-RLCshaped}, there exist a Hamiltonian $(s, p)$-path $P_1$, a Hamiltonian $(q, t)$-path $P_3$, and a Hamiltonian $(w, z)$-path $P_2$ of $R_{12}$, $R_{11}$, and $R_2$, respectively (see Fig. \ref{Fig_HP1}(d)). Then, $P = P_1 \Rightarrow P_2\Rightarrow P_3$ forms a Hamiltonian $(s, t)$-path of $O(m,n; k,l; a,b,c,d)$, as depicted in Fig. \ref{Fig_HP1}(e).

\begin{figure}[h]
\centering
\includegraphics[scale=0.85]{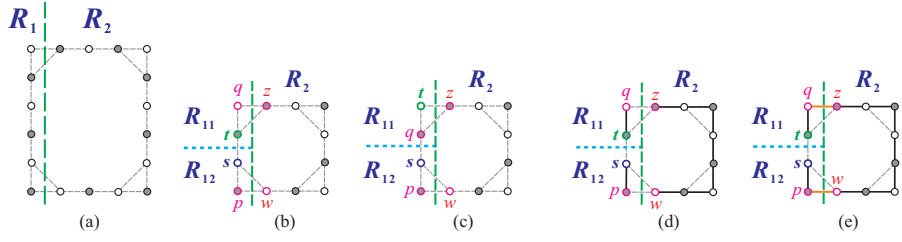}
\caption{(a) A vertical separation on $O(m,n; k,l; a,b,c,d)$ under that $s_x, t_x\leqslant a$ and $a=b=c=d=1$, (b)--(c) a horizontal separation on $R_1$, (d) a Hamiltonian path in $R_{11}$, $R_{12}$, and $R_2$, and (e) a Hamiltonian $(s,t)$-path in $O(m,n; k,l; a,b,c,d)$, where (b) indicates  $s\thicksim t$, (c) indicates $s\not\thicksim t$, and bold lines indicate the constructed Hamiltonian $(s, t)$-path.} \label{Fig_HP1}
\end{figure}

Case 2: $a\geqslant 2$ and $c=1$. Depending on the size of $a$, we consider the following two subcases:

\hspace{0.5cm}Case 2.1: $a = 2$. In this subcase, we first make a vertical separation on $O(m,n; k,l; a,b,c,d)$ to obtain two disjoint supergrid subgraphs $R_1=R(m_1,n)$ and $R_2 = C(m-m_1,n; k,l; c,d)$, where $m_1=a$ (see Fig. \ref{Fig_HP2}(a)). Depending on whether $s, t\in \{(m_1, 1), (m_1, n)\}$, we consider the following subcases.

\begin{figure}[h]
\centering
\includegraphics[scale=0.85]{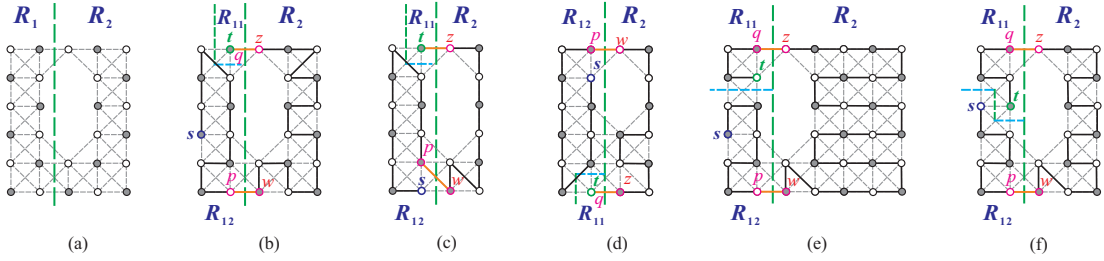}
\caption{(a) A vertical separation on $O(m,n; k,l; a,b,c,d)$ under that $s_x, t_x\leqslant a$, $a = 2$, and $c=1$ to obtain $R_1$ and $R_2$, (b), (c), (d), (f) a vertical and horizontal separations on $R_1$, and (e) a horizontal separation on $R_1$, where bold lines indicate the constructed Hamiltonian path.} \label{Fig_HP2}
\end{figure}

\hspace{1cm}Case 2.1.1: $s$ or $t\in\{(m_1, 1), (m_1, n)\}$. Without loss of generality, assume that $t\in \{(m_1, 1), (m_1, n)\}$.

\hspace{1.5cm}Case 2.1.1.1: $t=(m_1, 1)$. We make a vertical and horizontal separations on $R_1$ to obtain two disjoint supergrid subgraphs $R_{11}$ and $R_{12}$ such that $R_{11} = R(1, t_y)$ and $R_{12} = L(m_1,n; 1,t_y)$ (see Fig. \ref{Fig_HP2}(b) and Fig. \ref{Fig_HP2}(c)). Let $q\in V(R_{11})$, $p\in V(R_{12})$, and $w, z\in V(R_2)$ such that such that $p\thicksim w$, $q\thicksim z$, $z = (m_1+1, 1)$, $w = (m_1+1, n)$, $q = (m_1, 1) = t$, and $p = (m_1, n)$ if $ s\neq (m_1, n)$; otherwise $p = (m_1, n-1)$. Consider $(R_2, w, z)$. Since $z_x = w_x = m_1+1 = a+1$, $z_y = 1$, and $w_y =n $, clearly $(R_2, w, z)$ does not satisfy conditions (F1)--(F2) and (F4)--(F6). Consider $(R_{12}, s, p)$. Since $(p_y = n$ and $s_y \leqslant n)$ or $(p_y = n-1$ and $s_y = n)$, a simple check shows that $(R_{12}, s, p)$ does not satisfy conditions (F1), (F2), and (F3). A Hamiltonian $(s, t)$-path of $O(m,n; k,l; a,b,c,d)$ can be constructed by similar to Case 1. For instance, Figs. \ref{Fig_HP2}(b)--(c) depict the constructed Hamiltonian $(s, t)$-paths of $O(m,n; k,l; a,b,c,d)$ in this subcase.

\hspace{1.5cm}Case 2.1.1.2: $t=(m_1, n)$. If $d = 1$, then by symmetry a Hamiltonian $(s, t)$-path of $O(m,n; k,l; a,b,c,d)$ can be constructed by similar to Case 2.1.1.1, where $R_{11}=R(1, 1)$, $R_{12}=L(m_1,n; 1,1)$, $q=(m_1, n)$, $p=(m_1, 1)$, $w=(m_1+1, 1)$, and $z=(m_1+1, n)$ (see Fig. \ref{Fig_HP2}(d)). Consider $d\geqslant 2$. We make a vertical and horizontal separations on $R_1$ to obtain two disjoint supergrid subgraphs $R_{11}$ and $R_{12}$ such that $R_{11} = R(1, 1)$ and $R_2 = L(n,m_1; 1,1)$ (see Fig. \ref{Fig_HP2}(d)). Let $q\in V(R_{11})$, $p\in V(R_{12})$, and $w, z\in V(R_2)$ such that such that $p\thicksim w$, $q\thicksim z$, $z = (m_1+1, n)$, $w = (m_1+1, 1)$, $q = (m_1, n) = t$, and $p = (m_1, 1)$ if $ s\neq (m_1, 1)$; otherwise $p = (m_1, 2)$. Then, a Hamiltonian $(s, t)$-path of $O(m,n; k,l; a,b,c,d)$ can be constructed by similar to Case 1. For instance, Fig. \ref{Fig_HP2}(d) depicts the constructed Hamiltonian $(s, t)$-path of $O(m,n; k,l; a,b,c,d)$ in this subcase.

\hspace{1cm}Case 2.1.2: $s, t\notin \{(m_1, 1),(m_1, n)\}$. Without loss of generality, assume that $t_y\leqslant s_y$. First, let $s_y\neq t_y$. Then a Hamiltonian $(s, t)$-path of $O(m,n; k,l; a,b,c,d)$ can be constructed by similar to Case 1, where $n_1 = t_y$ (see Fig. \ref{Fig_HP2}(e)). Now, let $s_y = t_y$. Then a Hamiltonian $(s, t)$-path of $O(m,n; k,l; a,b,c,d)$ can be constructed by similar to Case 2.1.1.1, where $R_{12} = L(m_1,n-n_1+1; 1,1)$, $R_{11} = L(m_1,n_1; 1,1)$, and $n_1=s_y$ (see Fig. \ref{Fig_HP2}(f)). For example, Figs. \ref{Fig_HP2}(e)--(f) show the constructed Hamiltonian $(s, t)$-paths of $O(m,n; k,l; a,b,c,d)$ in this subcase.

\hspace{0.5cm}Case 2.2: $a \geqslant 3$. We make a vertical separation on $O(m,n; k,l; a,b,c,d)$ to obtain two disjoint supergrid subgraphs $R_1 =R(m_1, n)$ and $R_2 = O(m-m_1,n; k,l; 1,b,c,d)$, where $m_1 = a-1$ (see Fig. \ref{Fig_HP3}(a)). Depending on the positions of $s$ and $t$, there are the following three subcases:

\begin{figure}[h]
\centering
\includegraphics[scale=0.85]{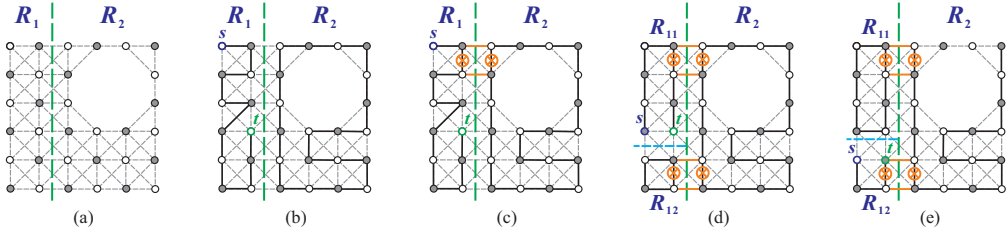}
\caption{(a) A vertical separation on $O(m,n; k,l; a,b,c,d)$ under that $s_x, t_x\leqslant a$ and $a\geqslant 3$, (b) a Hamiltonian $(s,t)$-path in $R_1$ and a Hamiltonian cycle in $R_2$, (c) a Hamiltonian $(s, t)$-path in $O(m,n; k,l; a,b,c,d)$ for $\{s, t\}$ is not a vertex cut of $R_{11}$, (d) and (e) horizontal separations on $R_1$ and Hamiltonian $(s, t)$-paths of $O(m,n; k,l; a,b,c,d)$ for $\{s, t\}$ is a vertex cut of $R_{11}$, where (b)--(e) $s, t\in R_1$, bold lines indicate the constructed Hamiltonian path, and $\otimes$ represents the destruction of an edge while constructing a such Hamiltonian path.} \label{Fig_HP3}
\end{figure}

\hspace{1.0cm}Case 2.2.1: $s, t\in R_1$.

\hspace{1.5cm}Case 2.2.1.1: $(m_1\geqslant 3)$ or $(m_1=2$ and $[(s_y\neq t_y)$, $(s_y=t_y=1)$, or $(s_y=t_y=n)])$. In this subcase, $\{s, t\}$ is not a vertex cut of $R_1$. Consider $(R_1, s, t)$. It is easy to check that $(R_1, s, t)$ does not satisfy condition (F1). Hence, by Lemmas \ref{HamiltonianConnected-Rectangular-wz_rectangle}--\ref{HP-3rectangle-boundary_path}, $R_1$ contains a Hamiltonian $(s, t)$-path $P_1$ in which one edge $e_1$ is placed to face $R_2$. By Theorem \ref{HC-Oshaped}, $R_2$ contains a Hamiltonian cycle $HC_2$ such that one flat face is placed to face $R_1$. Then, there exist two edges $e_1 \in P_1$ and $e_2\in HC_2$ such that $e_1 \thickapprox e_2$ (see Fig. \ref{Fig_HP3}(b)). By Statement (2) of Proposition \ref{Pro_Obs}, $P_1$ and $HC_2$ can be combined into a Hamiltonian $(s, t)$-path of $O(m, n; k,l; a,b,c,d)$. The construction of a such Hamiltonian path is depicted in Fig. \ref{Fig_HP3}(c).

\hspace{1.5cm}Case 2.2.1.2: $m_1 = 2$ and $2\leqslant s_y=t_y\leqslant n-1$. In this subcase, $\{s, t\}$ is a vertex cut of $R_1$. We make a horizontal separation on $R_1$ to obtain two disjoint supergrid subgraphs $R_{11} =R(m_1, n_1)$ and $R_{12} = R(m_1, n-n_1)$ such that $n_1=s_y$ if $s_y\neq n-1$; otherwise $n_1=s_y-1$ (see Fig. \ref{Fig_HP3}(d) and \ref{Fig_HP3}(e)). Notice that if $n_1=s_y$, then $s, t\in R_{11}$; otherwise $s, t\in R_{12}$. Without loss of generality, assume that $s, t\in R_{11}$. Clearly since $s_y=t_y=n_1$, $(R_{11}, s, t)$ does not satisfy condition (F1). By Lemma \ref{HamiltonianConnected-Rectangular-wz_rectangle}--\ref{HP-3rectangle-boundary_path}, $R_{11}$ contains a Hamiltonian $(s,t)$-path $P_{11}$ in which one edge $e_{11}$ is placed to face $R_2$. By Lemma \ref{HC-rectangular_supergrid_graphs} and Theorem \ref{HC-Oshaped}, $R_{12}$ and $R_2$ contain Hamiltonian cycle $HC_{12}$ and $HC_2$, respectively. Then, there exist four edges $e_1, e_2\in  HC_2$, $e_{11}\in P_{11}$, and $e_{12}\in HC_{12}$ such that $e_1\thickapprox e_{11}$ and $e_2\thickapprox e_{12}$; as shown in Fig. \ref{Fig_HP3}(d). By Statements (1) and (2) of Proposition \ref{Pro_Obs}, $P_{11}$, $HC_{12}$, and $HC_2$ can be combined into a Hamiltonian $(s, t)$-path of $O(m,n; k,l; a,b,c,d)$. The construction of a such Hamiltonian path is depicted in Fig. \ref{Fig_HP3}(d). For the case of $s, t\in R_{12}$, a Hamiltonian $(s, t)$-path of $O(m,n; k,l; a,b,c,d)$ can be constructed by the same arguments, as shown in Fig. \ref{Fig_HP3}(e).

\hspace{1.0cm}Case 2.2.2: $s, t\in R_2$. Since $s_x, t_x\leqslant a$, thus $s_x=t_x=a$. Without loss of generality, assume that $t_y < s_y$. We make a vertical and horizontal separations on $R_2$ to obtain three disjoint supergrid subgraphs $R_{21} =R(1, t_y)$, $R_{22} = R(1,n -t_y)$, and $R_{23} =C(m-m_1-1,n; k,l; c,d)$. Let $R_\textrm{a} = R_1\cup R_{22} = L(a,n; 1,t_y)$. A Hamiltonian $(s, t)$-path of $O(m,n; k,l; a,b,c,d)$ can be constructed by similar to Case 2.1.1.1, where $p\in V(R_\textrm{a})$, $q \in V(R_{21})$, and $w, z\in V(R_{23})$ (see Fig. \ref{Fig_HP4} (a)--(c)).

\begin{figure}[h]
\centering
\includegraphics[scale=0.85]{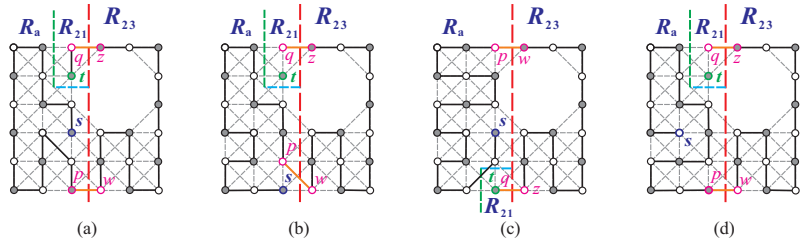}
\caption{A Hamiltonian $(s, t)$-path of $O(m,n; k,l; a,b,c,d)$ under that $a\geqslant 3$ and $c=1$ for (a)--(c) $s, t\in R_2$ and (d) $s\in R_1$ and $t\in R_2$.} \label{Fig_HP4}
\end{figure}

\hspace{1.0cm}Case 2.2.3: $s\in R_1$ and $t\in R_2$. A Hamiltonian $(s, t)$-path of $O(m,n; k,l; a,b,c,d)$ can be constructed by similar to Case 2.2.2. The construction of such a Hamiltonian $(s, t)$-path is shown in Fig. \ref{Fig_HP4}(d).

Case 3: $a, b, c, d\geqslant 2$. For the case of $a = 2$, a Hamiltonian $(s, t)$-path of $O(m,n; k,l; a,b,c,d)$ can be constructed by the same arguments in Case 2.1. And for $a \geqslant 3$, a Hamiltonian $(s, t)$-path of $O(m,n; k,l; a,b,c,d)$ can be obtained by the same construction in Case 2.2.
\end{proof}

\begin{lem}\label{HP-Oshaped2}
Let $O(m,n; k,l; a,b,c,d)$ be an $O$-shaped supergrid graph, and let $s$ and $t$ be its two distinct vertices such that $s_x\leqslant a$, $t_x\geqslant a+1$, and $(O(m,n; k,l; a,b,c, d), s, t)$ does not satisfy conditions $\mathrm{(F1)}$ and $\mathrm{(F10)}$--$\mathrm{(F13)}$. Then, $O(m,n; k,l; a,b,c,d)$ contains a Hamiltonian $(s, t)$-path, i.e., $HP(O(m,n; k,l; a,b,c,d), s, t)$ does exist.
\end{lem}
\begin{proof}
Let $a=b=c=d=1$. Then $t_y = s_y$, $s\in \{(1,1), (1,n)\}$, and $[(t_x\leqslant a+2)$ or $(k=2$ and $t\in \{(m,1), (m,n)\})]$. If $k\geqslant 3$, $t_x\geqslant a+3$, and $[(s\in \{(1,1), (1,n)\}$ and $t_y\neq s_y)$ or $(s\notin \{(1,1), (1,n)\})]$, then $(O(m,n; k,l; a,b,c,d), s, t)$ satisfies condition (F1) or (F10). So, $s_y = t_y = 1$ or $s_y = t_y = n$, and hence this case is isomorphic to Case 1 of Lemma \ref{HP-Oshaped1}. Therefore, in the following cases we assume that $a\geqslant 2$. Also for the case $b=c=d=1$, without loss of generality, assume that $s_y, t_y\leqslant c+l$. Consider the following cases:

Case 1: $c = 1$. In this case, $a\geqslant 2$ and $c=1$, and there are four subcases based on the sizes of $b$, $c$, and $d$ (see Fig. \ref{Fig_CaseII}). Depending on the location of $t$, we consider the following subcases:

\hspace{0.5cm}Case 1.1: $a+1\leqslant t_x\leqslant a+k$ and $t_y=1$. In this subcase, $(t_x\leqslant a+2)$ or $ (b, d\geqslant 2$ and $t_x\geqslant a+k-1)$. If ($b=1$ or $d=1$, $k\geqslant 3$, and $t_x\geqslant a+3$) or ($b, d\geqslant 2$, $k\geqslant 5$, and $a+3\leqslant t_x\leqslant a+k-2$), then $(O(m,n; k,l; a,b,c,d), s, t)$ satisfies (F11)--(F13). Note that ($b=1$ or $d=1$, $k\geqslant 3$, and $t_x\geqslant a+3$) satisfies condition (F12) or (F13), and ($b, d\geqslant 2$, $k\geqslant 5$, and $a+3\leqslant t_x\leqslant a+k-2$) satisfies condition (F11). We then have the following subcases:

\begin{figure}[h]
\centering
\includegraphics[scale=0.85]{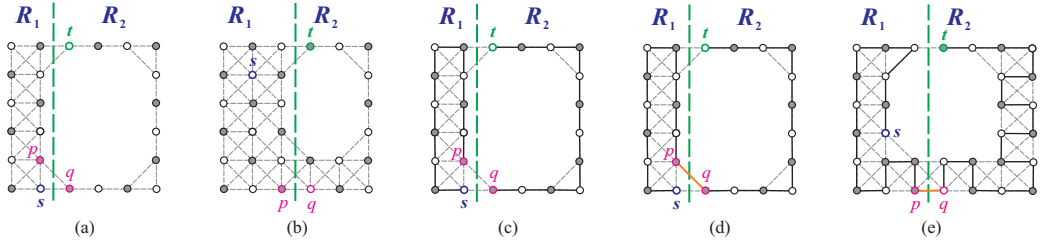}
\caption{(a) and (b) A vertical separation on $O(m,n; k,l; a,b,c,d)$ under that $a\geqslant 2$, $c=1$, and $t = (a+1, 1)$, (c) a Hamiltonian $(s, p)$-path in $R_1$ and a Hamiltonian $(q, t)$-path in $R_2$ for (a), (d) a Hamiltonian $(s, t)$-path in $O(m,n; k,l; a,b,c,d)$ for (a), and (e) a vertical separation on $O(m,n; k,l; a,b,c,d)$ and a Hamiltonian $(s, t)$-path in $O(m,n; k,l; a,b,c,d)$ under that $a\geqslant 2$, $c=1$, and $t = (a+2, 1)$, where bold lines indicate the constructed Hamiltonian path.} \label{Fig_HP5}
\end{figure}

\hspace{1.0cm}Case 1.1.1: $t_x = a+1$. We first make a vertical separation on $O(m,n; k,l; a,b,c,d)$ to obtain two disjoint supergrid subgraphs $R_1 = R(m_1, n)$ and $R_2 = C(m-m_1,n; k,l; c,d)$, where $m_1 = a$ (see Fig. \ref{Fig_HP5}(a) and \ref{Fig_HP5}(b)). Let $p\in V(R_1)$ and $q \in V(R_{2})$ such that $p\thicksim q$, $q = (m_1+1, n)$, and $p = (m_1, n)$ if $s\neq (m_1, n)$; otherwise $p = (m_1, n-1)$. Consider $(R_2, q, t)$. Since $t_y = 1$ and $q_y = n$, clearly $(R_2, q, t)$ does not satisfy (F1), (F2) and (F4)--(F6). Consider $(R_1, s, p)$. Condition (F1) holds, if $m_1 = 2$ and $s_y = p_y = n-1$. Clearly, it contradicts that $p = (m_1, n)$ when $s \neq (m_1, n)$. Thus, $(R_1, s, p)$ does not satisfy condition (F1). Since $(R_1, s, p)$ and $(R_2, q, t)$ do not satisfy conditions (F1), (F2), and (F4)--(F6), by Theorem \ref{HP-Theorem-RLCshaped}, there exist a Hamiltonian $(s, p)$-path $P_1$ and a Hamiltonian $(q, t)$-path $P_2$ of $R_1$ and $R_2$, respectively (see Fig. \ref{Fig_HP5}(c)). Then, $P = P_1 \Rightarrow P_2$ forms a Hamiltonian $(s, t)$-path of $O(m,n; k,l; a,b,c,d)$, as depicted in Fig. \ref{Fig_HP5}(d).

\hspace{1.0cm}Case 1.1.2: $t_x = a+2$. In this subcase, $k\geqslant 2$. A Hamiltonian $(s, t)$-path of $O(m,n; k,l; a,b,c,d)$ can be constructed by similar to Case 1.1.1, where $R_1 = C(m_1,n; 1,l; c,d)$, $R_2 = C(m-m_1,n; k-1,l; c,d)$, and $m_1 = a+1$ (see Fig. \ref{Fig_HP5}(e)). Fig. \ref{Fig_HP5}(e) also depicts the constructed Hamiltonian $(s, t)$-path of $O(m,n; k,l; a,b,c,d)$ in this subcase.

\hspace{1.0cm}Case 1.1.3: $b, d\geqslant 2$ and $t_x\geqslant a+k-1$. In this subcase, $a+k\geqslant t_x\geqslant a+k-1$. Thus, either $t = (a+k, 1)$ or $t = (a+k-1, 1)$.

\hspace{1.5cm}Case 1.1.3.1: $t = (a+k, 1)$. In this subcase, $k\geqslant 3$. We make two vertical and one horizontal separations on $O(m,n; k,l; a,b,c,d)$ to obtain two disjoint supergrid subgraphs $R_1 = C(n,m; l+c,k; a,b)$ and $R_2 = R(k, c)$; as shown in Fig. \ref{Fig_HP6}(a). Let $p\in V(R_1)$ and $q\in V(R_2)$ such that $p\thicksim q$, $q = (a+1, 1)$, and $p = (a, 1)$ if $s\neq (a,1 )$; otherwise $p = (a, 2)$. Consider $(R_1, s, p)$. Since $a, b, d\geqslant 2$, it is enough to show that $(R_1, s, p)$ is not in condition (F1). Condition (F1) holds, if $s_y = p_y = 2$. Clearly, it contradicts that $p = (a, 1)$ when $s \neq (a, 1)$. Consider $(R_2, q, t)$. Since $q = (a+1, 1)$ and $t = (a+k, 1)$, it is clear that $(R_2, q, t)$ does not satisfy condition (F1). A Hamiltonian $(s, t)$-path of $(O(m,n; k,l; a,b,c,d)$ can be constructed by similar to Case 1.1.1. Fig. \ref{Fig_HP6}(a) depicts such a constructed Hamiltonian $(s, t)$-path of $O(m,n; k,l; a,b,c,d)$.

\begin{figure}[h]
\centering
\includegraphics[scale=0.85]{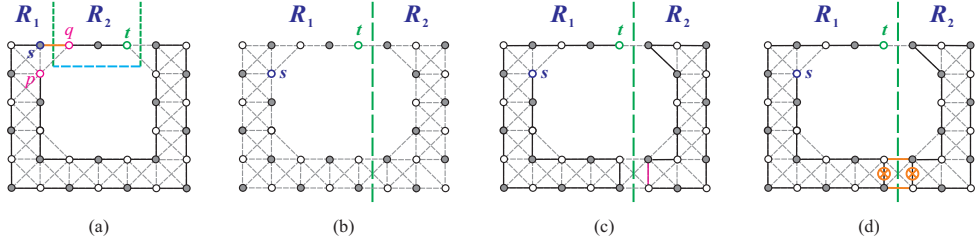}
\caption{(a) Two vertical and one horizontal separations on $O(m,n; k,l; a,b,c,d)$ under that $a,b,d\geqslant 2$, $c=1$, and $t = (a+k, 1)$, (b) a vertical separation on $O(m,n; k,l; a,b,c,d)$ under that $a,b,d\geqslant 2$, $c=1$, and $t = (a+k-1, 1)$, (c) a Hamiltonian $(s, t)$-path in $R_1$ and a Hamiltonian cycle in $R_2$ for (b), and (d) a Hamiltonian $(s, t)$-path in $O(m,n; k,l; a,b,c,d)$ for (b), where bold lines indicate the constructed Hamiltonian path and $\otimes$ represents the destruction of an edge while constructing a Hamiltonian $(s, t)$-path of $O(m,n; k,l; a,b,c,d)$.} \label{Fig_HP6}
\end{figure}

\hspace{1.5cm}Case 1.1.3.2: $t = (a+k-1, 1)$. In this subcase, $k\geqslant 4$. We make a vertical separation on $O(m,n; k,l; a,b,c,d)$ to obtain two disjoint supergrid subgraphs $R_1 = C(t_x,n; t_x-a,l; c,d)$ and $R_2 = C(m-t_x,n; 1,l; c,d)$ (see Fig. \ref{Fig_HP6}(b)). Consider $(R_1, s, t)$. Since $a, d\geqslant 2$, $t = (a+k-1,1)$ and $s_x\leqslant a$, it is cleat that $(R_1, s, t)$ does not satisfy (F1), (F2), and (F4)--(F6). Since $(R_1, s, t)$ does not satisfy conditions (F1), (F2), and (F4)--(F6), by Theorem \ref{HP-Theorem-RLCshaped}, $R_1$ contains a Hamiltonian $(s, t)$-path. Using the algorithm of \cite{Keshavarz19b}, we can construct a Hamiltonian $(s, t)$-path $P_1$ of $R_1$ in which one edge $e_1$ is placed to face $R_2$. By Theorem \ref{HC-LCshaped}, $R_2$ contains a Hamiltonian cycle $HC_2$. Note that by the construction of Hamiltonian cycle in \cite{Keshavarz19b}, we can construct $HC_2$ such that its one flat face is placed to $R_2$ (see Fig. \ref{Fig_HP6}(c)). Then, there exist two edges $e_1 \in P_1$ and $e_2\in HC_2$ such that $e_1 \thickapprox e_2$ (see Fig. \ref{Fig_HP6}(c)). By Statement (2) of Proposition \ref{Pro_Obs}, $P_1$ and $HC_2$ can be combined into a Hamiltonian $(s, t)$-path of $O(m,n; k,l; a,b,c,d)$. The construction of a such Hamiltonian path is depicted in Fig. \ref{Fig_HP6}(d).

\hspace{0.5cm}Case 1.2: $b=1$ and $t_y\leqslant c+l$. In this case, ($k\leqslant 2$ or $l\leqslant 2$, and $t=(m, 1)$) or ($d\geqslant 2$ and $t_y\geqslant c+l-1$). If ($l\geqslant 3$ and $2\leqslant t_y\leqslant c+l-2$) or ($l, k\geqslant 3$ and $t = (m, 1)$), then $(O(m,n; k,l; a,b,c,d), s, t)$ satisfies condition (F13).

\hspace{1cm}Case 1.2.1: $t =( m, 1)$. In this subcase, $k\leqslant 2$ or $l\leqslant 2$.

\hspace{1.5cm}Case 1.2.1.1: $k\leqslant 2$. A Hamiltonian $(s, t)$-path of $O(m,n; k,l; a,b,c,d)$ can be constructed by similar to Case 1.1.1, where $k = 1$, $R_1 = R(m_1, n)$, $R_2 = C(m-m_1,n; k,l; c,d)$, and $m_1 = a$ (see Fig. \ref{Fig_HP7}(a)) or $k = 2$, $R_1 = C(a+1,n; 1,l; c,d)$ and $R_2 = C(m-a-1,n; 1,l; c,d)$ (see Fig. \ref{Fig_HP7}(b)). Fig. \ref{Fig_HP7}(a) and Fig. \ref{Fig_HP7}(b) also show the constructed Hamiltonian $(s, t)$-paths of $O(m,n; k,l; a,b,c,d)$.

\begin{figure}[h]
\centering
\includegraphics[scale=0.85]{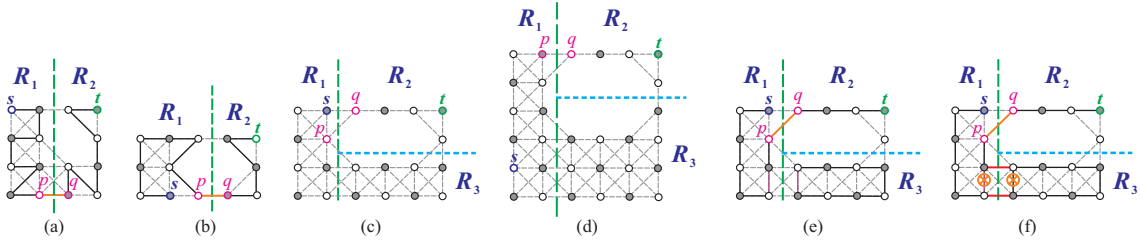}
\caption{(a) and (b) A vertical separation on $O(m,n; k,l; a,b,c,d)$ under that $a\geqslant 2$, $b=c=1$, $t = (m, 1)$, and $k\leqslant 2$, (c) and (d) a vertical and horizontal separations on $O(m,n; k,l; a,b,c,d)$ under that $a\geqslant 2$, $b=c=1$, $t = (m, 1)$, and $k\geqslant 3$, $l\leqslant 2$, (e) a Hamiltonian $(s, t)$-path in $R_1\cup R_2$ and a Hamiltonian cycle in $R_3$ for (c), and (f) a Hamiltonian $(s, t)$-path in $O(m,n; k,l; a,b,c,d)$ for (c).} \label{Fig_HP7}
\end{figure}

\hspace{1.5cm}Case 1.2.1.2: $k\geqslant 3$ and $l\leqslant 2$. We make a vertical and horizontal separations on $O(m,n; k,l; a,b,c,d)$ to obtain three disjoint supergrid subgraphs $R_1 = R(a, n)$, $R_2 = L(m-a,n_1; k,n_1-c)$, and $R_3 = R(m-a, d)$ if $l=1$; otherwise $R_3=L(m-a, n-n_1; k, l-n_1+c)$, where $n_1 = t_y+1 = 2$ (see Fig. \ref{Fig_HP7}(c) and \ref{Fig_HP7}(d)). Let $p\in V(R_1)$ and $q\in V(R_2)$ such that $p\thicksim q$, $q = (a+1, 1)$, and $p = (a, 1)$ if $s\neq (a, 1)$; otherwise $p = (a, 2)$. Consider $(R_1, s, p)$. Condition (F1) holds, if $a = 2$ and $s_y = p_y  =2$. Clearly, it contradicts that $p = (a, 1)$ when $s \neq (a,1)$. Now, consider $(R_2, q, t)$. Since $b = c = 1$, $q = (a+1, 1)$, and $t = (m, 1)$, it is easy to check that $(R_2, q, t)$ does not satisfy (F1), (F2), and (F3). Since $(R_1, s, p)$ and $(R_2, q, t)$ do not satisfy conditions (F1), (F2), and (F2), by Theorem \ref{HP-Theorem-RLCshaped}, there exist a Hamiltonian $(s, p)$-path $P_1$ and a Hamiltonian $(q, t)$-path $P_2$ of $R_1$ and $R_2$, respectively. Note that $P_1$ is a canonical Hamiltonian path of $R_1$. Then, $P = P_1 \Rightarrow P_2$ forms a Hamiltonian $(s, t)$-path of $R_1\cup R_2$, as depicted in Fig. \ref{Fig_HP7}(e). By Lemma \ref{HC-rectangular_supergrid_graphs} or Theorem \ref{HC-LCshaped}, $R_3$ contains a Hamiltonian cycle $HC_3$. We can place one flat face of $HC_3$ to face $R_1$. Then, there exist two edges $e_1\in P$ and $e_3 \in HC_3$ and such that $e_1 \thickapprox e_3$ (see Fig. \ref{Fig_HP7}(f)). By Statement (2) of Proposition \ref{Pro_Obs}, $P$ and $HC_3$ can be combined into a Hamiltonian $(s, t)$-path of $O(m,n; k,l; a,b,c,d)$. The construction of a such Hamiltonian path is depicted in Fig. \ref{Fig_HP7}(f).

\hspace{1cm}Case 1.2.2: $t = (m, c+l)$. A Hamiltonian $(s, t)$-path of $O(m,n; k,l; a,b,c,d)$ can be constructed by similar to Case 1.1.3.1, where $R_1 = L(m,n; k+b, c+l)$ and $R_2 = L(k+b,c+l; k,l)$ (see Fig. \ref{Fig_HP8}(a) and \ref{Fig_HP8}(b)).

\begin{figure}[h]
\centering
\includegraphics[scale=0.85]{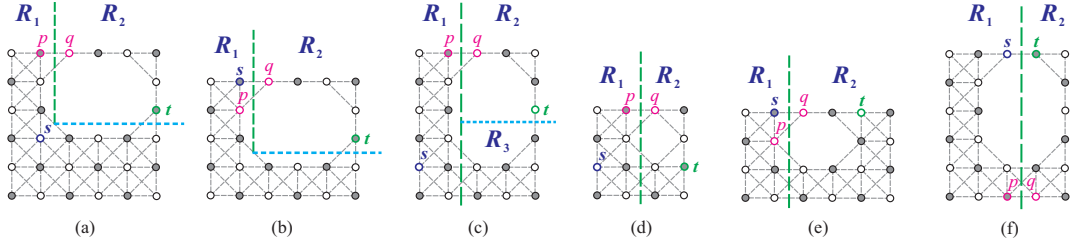}
\caption{(a)--(c) A vertical and horizontal separations on $O(m,n; k,l; a,b,c,d)$ under that $a\geqslant 2$, $b=c=1$, and $t = (m, c+l)$ or $t = (m, c+l-1)$, (d)--(e) a vertical separation on $O(m,n; k,l; a,b,c,d)$ for $a,d\geqslant 2$ and $c=1$, and (f) a vertical separation on $O(m,n; k,l; a,b,c,d)$ for $a\geqslant 2$, $c=1$, $s_y=t_y=1$ and $a+1\leqslant s_x,t_x\leqslant a+k$.} \label{Fig_HP8}
\end{figure}

\hspace{1cm}Case 1.2.3: $t = (m, c+l-1)$. In this subcase, $l\geqslant 2$. A Hamiltonian $(s, t)$-path of $O(m,n; k,l; a,b,c,d)$ can be constructed by similar to Case 1.2.1.2, where $n_1 = t_y$ (see Fig. \ref{Fig_HP8}(c)).

\hspace{.5cm}Case 1.3: $d\geqslant 2$ and $[(b\geqslant 2$ and $[(t_y\geqslant 2)$ or $(t_y=1$ and $t_x>a+k)])$ or $(b=1$ and $t_y\geqslant c+l+1)]$. A Hamiltonian $(s, t)$-path of $O(m,n; k,l; a,b,c,d)$ can be constructed by similar to Case 1.1.1, where $q = (m_1+1, 1)$, and $p = (m_1, 1)$ if $s\neq (m_1,1)$; otherwise $p=(m_1,2)$ (see Fig. \ref{Fig_HP8}(d) and \ref{Fig_HP8}(e)). Notice that since $d, b\geqslant 2$, $q = (m_1+1, 1)$, and $[(t_x > a+k)$ or $(t_x\leqslant a+k$ and $t_y > c+l)]$, it is easy to check that $(R_2, q, t)$ does not satisfy conditions (F1), (F2), and (F4)--(F6).

Case 2: $a,b,c,d\geqslant 2$. A Hamiltonian $(s, t)$-path of $O(m,n; k,l; a,b,c,d)$ can be constructed by similar to Case 1.1.1 (see Figs. \ref{Fig_HP5}(c)--(d)), where $m_1 = a$ and
$$\begin{cases}
  p = (m_1, n)\ \mathrm{and}\ q = (m_1+1, n),     & \mathrm{if}\ s\neq (m_1, n)\ \mathrm{and}\ t\neq (m_1+1, n);\\
  p = (m_1, 1)\ \mathrm{and}\ q = (m_1+1, 1),     & \mathrm{if}\ s = (m_1, n)\ \mathrm{and}\ t = (m_1+1, n);\\
  p = (m_1, n)\ \mathrm{and}\ q = (m_1+1, n-1),   & \mathrm{if}\ s\neq (m_1, n)\ \mathrm{and}\ t = (m_1+1, n);\\
  p = (m_1, n-1)\ \mathrm{and}\ q = (m_1+1, n),   & \mathrm{otherwise.}
\end{cases}$$

\noindent Consider $(R_2, q, t)$. Since $b,c,d\geqslant 2$, $q_x = m_1+1$, and $t_x\geqslant m_1+1$, it is clear that $(R_2, q, t)$ does not satisfy (F1), (F2), and (F4)--(F6). Now, consider $(R_1, s, p)$. Condition (F1) holds only if ($s_y=p_y = n-1$) or ($s_y=p_y = 2$). Obviously, it contradicts that $p = (m_1, n)$ when $s\neq (m_1, n)$ or $p = (m_1, 1)$ when $s\neq (m_1, 1)$. Thus, $(R_1, s, p)$ does not satisfy condition (F1).
\end{proof}

\begin{lem}\label{HP-Oshaped3}
Let $O(m,n; k,l; a,b,c,d)$ be an $O$-shaped supergrid graph, and let $s$ and $t$ be its two distinct vertices such that $s_x, t_x\geqslant a+1$ and $(O(m,n; k,l; a,b,c,d), s, t)$ does not satisfy conditions $\mathrm{(F1)}$ and $\mathrm{(F10)}$--$\mathrm{(F13)}$. Then, $O(m,n; k,l; a,b,c,d)$ contains a Hamiltonian $(s, t)$-path, i.e., $HP(O(m,n; k,l; a,b,c,d), s, t)$ does exist.
\end{lem}
\begin{proof}
In the following, we will assume that $a\geqslant 2$ and $c=1$. The cases of $(a=b=c=d=1)$ and $(a,b,c,d\geqslant 2)$ are isomorphic to Lemmas \ref{HP-Oshaped1} and \ref{HP-Oshaped2}. If $c=d=1$ and $s_x, t_x\leqslant a+k$, then $s_y=t_y=1$ or $s_y=t_y=n$. Notice that if $s_y = 1$ and $t_y = n$, then $(O(m,n; k,l; a,b,c,d), s, t)$ satisfies condition (F1), i.e. $\{s, t\}$ is a vertex cut of $O(m,n; k,l; a,b,c,d)$. Thus, without loss of generality, assume that $s_y=t_y=1$ when $c=d=1$ and $s_x, t_x\leqslant a+k$. Consider the following three cases:

Case 1: One of the following cases holds:

\hspace{1.15cm}(1) $b\geqslant 2$ and $[(s_x,t_x > a+k)$ or $(s_x\leqslant a+k$ and $t_x > a+k)]$; or

\hspace{1.15cm}(2) $d\geqslant 2$ and $s_y > c+l$ or $t_y > c+l$.

\noindent In these cases, assume that $(1, 1)$ is the coordinates of vertex in upper-right corner when $b\geqslant 2$, or down-right corner when $d\geqslant 2$, in $O(m,n; k,l ; a,b,c,d)$. Then, we can construct a Hamiltonian $(s, t)$-path of $(O(m,n; k,l; a,b,c,d)$ with the same arguments as we did in the proofs of Lemmas \ref{HP-Oshaped1} and \ref{HP-Oshaped2}. Note that these cases are isomorphic to the assumptions of Lemmas \ref{HP-Oshaped1} and \ref{HP-Oshaped2}.

Case 2: $s_y=t_y = 1$ and $s_x,t_x\leqslant a+k$. In this case, $s\thicksim t$. If $s\not\thicksim t$, then $O(m,n; k,l; a,b,c,d), s, t)$ satisfies condition (F1). A Hamiltonian $(s, t)$-path of $O(m,n; k,l; a,b,c,d)$ can be constructed by similar to Case 1.1.1 of Lemma \ref{HP-Oshaped2}, where $R_1 = C(m_1,n; m_1-a,l; c,d)$, $R_2 = C(m-m_1,n; a+k-m_1,l; c,d)$, and $m_1=s_x$ (see Fig. \ref{Fig_HP8}(f)).

Case 3: $b=c=1$, $[(d=1)$ or $(d\geqslant 2$ and $s_y,t_y\leqslant c+l)]$, and $[(s_x=t_x = m)$ or $(s_x\leqslant m-1$ and $t_x = m)]$.
In this case, assume that $(1, 1)$ is the coordinates of vertex in upper-right corner in $O(m,n; k,l; a,b,c,d)$. Then, a Hamiltonian $(s, t)$-path of $O(m,n; k,l; a,b,c,d)$ can be constructed by similar  Lemmas \ref{HP-Oshaped1}, when $s_x=t_x=m$, and Lemma \ref{HP-Oshaped2}, when $s_x\leqslant a+k(=m-1)$ and $t_x = m$.
\end{proof}

It follows from Lemma \ref{Necessary-condition-Oshaped} and Lemmas \ref{HP-Oshaped1}--\ref{HP-Oshaped3} that the following theorem shows the Hamiltonian connectivity of $O$-shaped supergrid graphs.

\begin{thm}\label{HP-Theorem-Oshaped}
Let $O(m,n; k,l; a,b,c,d)$ be an $O$-shaped supergrid graph, and let $s$ and $t$ be its two distinct vertices. Then, $(O(m,n; k,l; a,b,c,d), s, t)$ contains a Hamiltonian $(s, t)$-path if and only if $(O(m,n; k,l; a,b,c,d), s, t)$ does not satisfy conditions $\mathrm{(F1)}$ and $\mathrm{(F10)}$--$\mathrm{(F13)}$.
\end{thm}

\section{The Longest $(s, t)$-path Algorithm}\label{Sec_Algorithm}
From Theorem \ref{HP-Theorem-Oshaped}, we know that if $(O(m,n; k,l;a,b, c,d), s, t)$ satisfies one of conditions (F1) and (F10)--(F13), then $(O(m,n; k,l; a,b,c,d), s, t)$ contains no Hamiltonian $(s, t)$-path. So in this section, first for these cases we give upper bounds on the lengths of longest paths between $s$ and $t$. Then, we show that these upper bounds are equal to the lengths of longest $(s, t)$-paths in $O(m,n; k,l;a,b, c,d)$. Notice that the isomorphic cases are omitted, and assume that $s_x \leqslant t_x$. Then, we can only consider the cases of $a=b=c=d=1$, $a\geqslant 2$ and $c=1$, and $a,b,c,d\geqslant 2$ (see Section \ref{Sec_Preliminaries}). In the following, we use $\hat{L}(G, s, t)$ to denote the length of longest paths between $s$ and $t$, and $\hat{U}(G,s,t)$ to indicate the upper bound on the length of longest paths between $s$ and $t$, where $G$ is a rectangular, $L$-shaped, or $C$-shaped supergrid graph. By the length of a path we mean the number of vertices of the path. The following lemmas give these upper bounds.

\begin{figure}[h]
\centering
\includegraphics[scale=0.85]{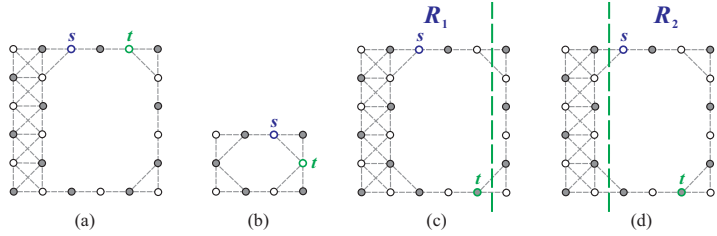}
\caption{The upper bound of the longest $(s, t)$-path under that $\{s, t\}$ is a vertex cut, where (a) (O1) holds, (b) (O2) holds, and (c)--(d) (O3) holds, where bold dash line indicates the vertical separation on $O(m,n; k,l; a,b,c,d)$.} \label{Fig-LongP1}
\end{figure}

We first consider the case of $\{s, t\}$ is a vertex cut of $O(m,n; k,l; a,b,c,d)$. We compute the upper bound of the longest $(s, t)$-path in this case as the following lemma.

\begin{lem}\label{Lemma:O1-O3}
Let $c = 1$ and $\{s, t\}$ be a vertex cut of $O(m,n; k,l;a,b, c,d)$. Then, the following conditions hold:
\begin{description}
  \item[$\mathrm{\textbf{(O1)}}$] If $k\geqslant 3$ and $s_y = t_y = 1$, then the length of any path between $s$ and $t$ cannot exceed $m\times n-k\times l-t_x+s_x+1$ (see Fig. $\mathrm{\ref{Fig-LongP1}(a)}$).
  \item[$\mathrm{\textbf{(O2)}}$] If $b = 1$, $a+1\leqslant s_x\leqslant a+k$, $s_y=1$, and $t_x=m$, then the length of any path between $s$ and $t$ cannot exceed $m\times n-k\times l-t_y-m+s_x+2$ (see Fig. $\mathrm{\ref{Fig-LongP1}(b)}$).
  \item[$\mathrm{\textbf{(O3)}}$] If $d = 1$, $a+1\leqslant s_x, t_x\leqslant a+k$, $s_y = 1$, and $t_y = n$, then the length of any path between $s$ and $t$ cannot exceed $\max\{\hat{L}(R_{1}, s, t), \hat{L}(R_{2}, s, t)\}$, where $R_1 = C(t_x,n; t_x-a,l, c, d)$, $R_2 = C(m-m_1,n; a+k-m_1,l, c, d)$, and $m_1 = s_x-1$ (see Fig. $\mathrm{\ref{Fig-LongP1}(c)}$ and $\mathrm{\ref{Fig-LongP1}(d)}$).
\end{description}
\end{lem}
\begin{proof}
Consider Fig. \ref{Fig-LongP1}. Removing $s$ and $t$ clearly disconnects $O(m,n; k,l;a,b, c,d)$ into two components $R_1$ and $R_2$. Thus, a simple path between $s$ and $t$ can only go through one of these components. Therefore, its length cannot exceed the size of the largest component. Notice that, for (O1) (resp., (O2)), the length of any path between $s$ and $t$ is equal to $\max\{t_x-s_x+1, m\times n-k\times l-t_x+s_x+1\}$ (resp., $\max\{t_y+m-s_x,m\times n-k\times l-t_y-m+s_x+2\}$). Since $a\times n+b\times n+d\times k>t_x-s_x+1$ (resp., $a\times n+d\times (k+b) > t_y+m-s_x$), it is obvious that the length of any path between $s$ and $t$ cannot exceed $m\times n-k\times l-t_x+s_x+1$ (resp., $m\times n-k\times l-t_y-m+s_x+2$).
\end{proof}

Next, we consider the case that $\{s, t\}$ is not a vertex cut of $O(m,n; k,l; a,b,c,d)$. In this case, $(O(m,n; k,l;$ $a,b,c,d), s, t)$ may satisfy condition (F10), (F11), (F12), or (F13). The following lemma shows the upper bound of the longest $(s, t)$-path under that $c = 1$ and $(O(m,n; k,l; a,b,c,d), s, t)$ satisfies conditions (F11)--(F13).

\begin{figure}[!t]
\centering
\includegraphics[scale=0.85]{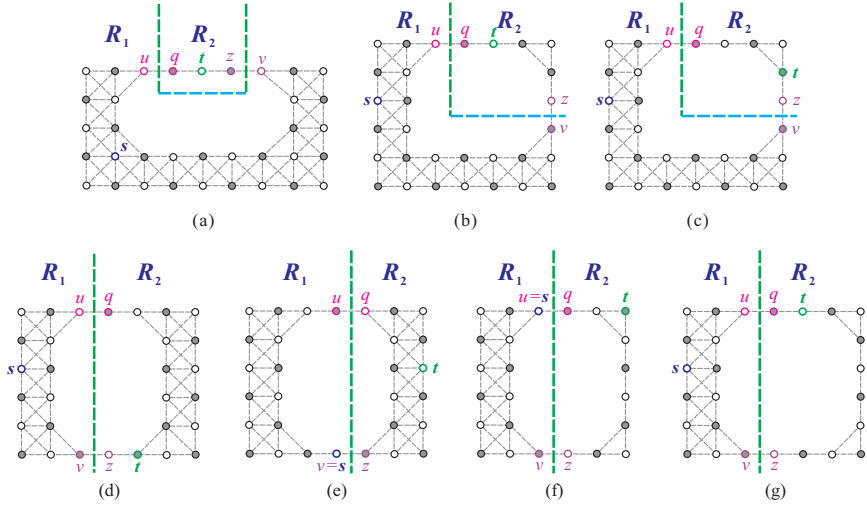}
\caption{The upper bound of the longest $(s, t)$-path in $O(m,n; k,l; a,b, c,d)$ for (a) condition (F11), (b) case (2.1) of condition (F13), (c) cases (2.2) and (2.3) of condition (F13), (d)--(e) condition (F11), and (f)--(g) cases (1.1), (1.2) and (3) of condition (F13), where bold dash line indicates the vertical or horizontal separation on $O(m,n; k,l; a,b,c,d)$.} \label{Fig-LongP2}
\end{figure}

\begin{lem}\label{Lemma:F11-F13}
Let $(O(m,n; k,l; a,b,c,d), s, t)$ satisfy $\mathrm{(F11)}$, $\mathrm{(F12)}$, or $\mathrm{(F13)}$ (cases $(1.1)$, $(1.2)$, $(2)$, and $(3)$ of $\mathrm{(F13)}$). Then, the following implications hold:
\begin{description}
  \item[$\mathrm{(1)}$] If $(O(m,n; k,l; a,b,c,d), s, t)$ satisfies condition $\mathrm{(F11)}$, then the length of any path between $s$ and $t$ cannot exceed $\max\{m \times n - k \times l - t_x+u_x+1, m \times n - k \times l -v_x+t_x+1\}$, where $u = (a+1, 1)$, $v = (a+k, 1)$, $t_y = 1$, and $a+3\leqslant t_x\leqslant a+k-2$ (see Fig. \emph{\ref{Fig-LongP2}(a)}).
  \item[$\mathrm{(2)}$] If $(O(m,n; k,l;a,b,c,d), s, t)$ satisfies case \emph{(2.1)} of condition $\mathrm{(F13)}$, then the length of any path between $s$ and $t$ cannot exceed $\max\{m \times n - k \times l - t_x+u_x+1, m\times n - k\times l - z_y-m+t_x+1\}$, where $u = (a+1, 1)$, $z = (m, c+l-1)$, $t_y = 1$, and $a+3\leqslant t_x\leqslant m-1$ (see Fig. \emph{\ref{Fig-LongP2}(b)}).
  \item[$\mathrm{(3)}$] If $(O(m,n; k,l; a,b, c,d), s, t)$ satisfies case \emph{(2.2)} or \emph{(2.3)} of condition $\mathrm{(F13)}$, then the length of any path between $s$ and $t$ cannot exceed $\max\{m \times n - k \times l - m+u_x-t_y+2, m\times n-k\times l-v_y+t_y+1\}$, where $u = (a+1, 1)$, $v = (m, c+l)$, $t_x = m$, and $1\leqslant t_y\leqslant l-1$ (see Fig. \emph{\ref{Fig-LongP2}(c)}).
  \item[$\mathrm{(4)}$] If $(O(m,n; k,l; a,b,c,d), s, t)$ satisfies condition $\mathrm{(F12)}$, case \emph{(1.1)} or \emph{(3)} of $\mathrm{(F13)}$, then the length of any path between $s$ and $t$ cannot exceed $\max\{\hat{L}(R_1, s, u) + \hat{L}(R_2, q, t), \hat{L}(R_1, s, v) + \hat{L}(R_2, z, t)\}$ (see Fig. \emph{\ref{Fig-LongP2}(d)--(g)}), where $u = (m_1, 1)$, $q = (m_1+1, 1)$, $v = (m_1, n)$, and $z = (m_1+1, n)$, $R_1 = C(m_1,n; m_1-a,l; c,d)$, $R_2 = C(m-m_1,n; a+k-m_1,l; c,d)$, if $k>1$; otherwise $R_2 = R(b, n)$, and $m_1 = a+1$ if $s_x\leqslant a$; otherwise $m_1 = s_x$.
\end{description}
\end{lem}
\begin{proof}
Consider Fig. \ref{Fig-LongP2}. It is clear that the longest $(s, t)$-path $P$ of $O(m,n; k,l; a,b,c,d)$ that starts from $s$ should pass through all (or some) the vertices of $R_1$, leaves $R_1$ at $u$ (or $v$), enters $R_2$ at $q$ (or $z$), and ends at $t$. Therefore, the length of any path between $s$ and $t$ cannot exceed $\max\{\hat{L}(R_1, s, u) + \hat{L}(R_2, q, t), \hat{L}(R_1, s, v) + \hat{L}(R_2, z, t)\}$.
\end{proof}

Finally, we consider condition (F10) and case (1.3) of condition (F13) as follows.

\begin{figure}[h]
\centering
\includegraphics[scale=.85]{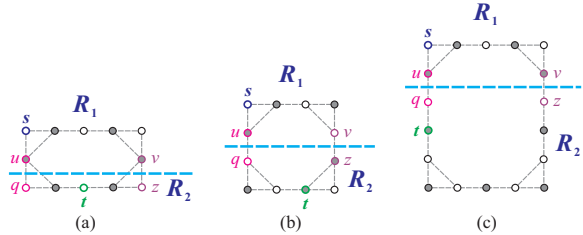}
\caption{The upper bound of the longest $(s, t)$-path in $O(m,n; k,l; a,b, c,d)$ for condition (O4), where (a)--(b) $d=1$, $t_y=n$, and $s_x\neq t_x$, and (c) $s_x=t_x$, $l\geqslant 3$, and $t_y\geqslant 4$, where bold dash line indicates the horizontal separation on $O(m,n; k,l; a,b,c,d)$.}\label{Fig-LongP3}
\end{figure}

\begin{lem}\label{Lemma:O4}
Let $a = c = 1$ and $s = (1, 1)$. If $(O(m,n; k,l; a,b,c,d), s, t)$ satisfies $\mathrm{(O4)}$, then the length of any path between $s$ and $t$ cannot exceed $\max\{\hat{L}(R_1, s, u) + \hat{L}(R_2, q, t), \hat{L}(R_1, s, v) + \hat{L}(R_2, z, t)\}$ (see Fig. $\mathrm{\ref{Fig-LongP3}}$ ), where $R_1 = C(n_1,m; 1,k; a,b)$, $R_2 = C(n-n_1,m; l-1,k; a,b)$ if $l > 1$; otherwise $R_2 = R(m, d)$, $n_1 = 1$, $u = (1, c+1)$, $q = (1, c+2)$, $v = (m, c+1)$, and $z = (m, c+2)$. Where condition $\mathrm{(O4)}$ is defined as follows:
\begin{description}
  \item[$\mathrm{\textbf{(O4)}}$] One of the following cases holds:
\begin{itemize}
  \item[\emph{(a)}] $d=1$, $t_y=n$, and $s_x\neq t_x$ (case \emph{(1)} of \emph{(F10)}); or
  \item[\emph{(b)}] $s_x=t_x$, $l\geqslant 3$, and $t_y\geqslant 4$ (case \emph{(2)} of \emph{(F10)} and case \emph{(1.3)} of \emph{(F13)}).
  \end{itemize}
\end{description}
\end{lem}
\begin{proof}
The proof is similar to the proof of Lemma \ref{Lemma:F11-F13}; see Figs. \ref{Fig-LongP3}.
\end{proof}

Let condition (O0) be defined as follows:

\begin{description}
  \item[$\mathrm{\textbf{(O0)}}$] $(O(m,n; k,l; a,b,c,d), s, t)$ does not satisfy any of conditions (F1), (F10), (F11), (F12), and (F13).
\end{description}

It is easy to check that any $(O(m,n; k,l; a,b,c,d), s, t)$ must satisfy one of conditions (O0), (O1), (O2), (O3), (O4), (F11), (F12), and (F13). If $(O(m,n; k,l; a,b,c,d), s , t)$ satisfies (O0), then $\hat{U}(O(m,n; k,l; a,b,c,d), s, t) = mn-kl$. Otherwise, $\hat{U}(O(m,n; k,l; a,b,c,d), s, t)$ can be computed using Lemma \ref{Lemma:O1-O3}--\ref{Lemma:O4}. We summarize them as follows, where $|G|=m\times n-k\times l$:\\

\small{
\noindent $\hat{U}(O(m,n; k,l; a,b,c,d), s, t)=
  \begin{cases}
    |G|-t_x+s_x+1,                                      &\mathrm{if \ (O1) \ holds;} \\
    |G|-t_y-m+s_x+2,                                    &\mathrm{if \ (O2)\ holds;} \\
    \max\{\hat{L}(R_1, s, t), \hat{L}(R_2, s, t)\},     &\mathrm{if \ (O3) \ holds;}\\
    \max\{|G|-t_x+u_x+1, |G|-v_x+t_x+1\},               &\mathrm{if \ (F11) \ holds;} \\
    \max\{|G|-t_x+u_x+1, |G|-z_y-m+t_x+1\},             &\mathrm{if \ case \ (2.1) \ of \ (F13) \ holds;} \\
    \max\{|G|-m+u_x-t_y+2, |G|-v_y+t_y+1\},             &\mathrm{if \ case \ (2.2) \ or \ (2.3) \ of \ (F13) \ holds;} \\
    \max\{\hat{L}(R_1, s, u) + \hat{L}(R_2, q, t), \hat{L}(R_1, s, v) + \hat{L}(R_2, z, t)\},   &\mathrm{if \ (O4),\ (F12),\ or \ (F13)\ (case \ 1.1\ or \ 3) \ holds;} \\
    mn-kl,                                              &\mathrm{if \ (O0) \ holds.}
  \end{cases}$}\\

Now, we show how to obtain a longest $(s,t)$-path for $O$-shaped supergrid graphs. Notice that if $(O(m,n; k,l; a,b,c,d), s, t)$ satisfies (O0), then, by Theorem \ref{HP-Theorem-Oshaped}, it contains a Hamiltonian $(s, t)$-path.

\begin{lem}\label{Lemma:long-Osupergrid}
If $(O(m,n; k,l; a,b c,d), s, t)$ satisfies one of conditions $\mathrm{(O1)}$--$\mathrm{(O4)}$ and $\mathrm{(F11)}$--$\mathrm{(F13)}$, then $\hat{L}(O(m,n; k,l; a,b,c,d), s, t) = \hat{U}(O(m,n; k,l; a,b,c,d), s, t)$.
\end{lem}
\begin{proof}
We prove this lemma by constructing a $(s, t)$-path $P$ such that its length equals to $\hat{U}(O(m,n; k,l; a,b,c,d), s, t)$. Consider the following cases:

Case 1: Conditions (O1) and (O2) hold. By Lemma \ref{Lemma:O1-O3}, $\hat{U}(O(m,n; k,l; a,b,c,d), s, t) = m\times n - k\times l - t_x+s_x+1$ and $\hat{U}(O(m,n; k,l; a,b,c,d), s, t) = m\times n - k\times l - t_y-m+s_x+2$, respectively. Consider Figs. \ref{Fig-LongP1}(a)--(b). We make a vertical separation on $O(m,n; k,l; a,b,c,d)$ to obtain two disjoint supergrid subgraphs $R_1 = C(m_1,n; m_1-a,l; c,d)$ and $R_2 = C(m-m_1,n; a+k-m_1,l)$ if $s_x\neq m-1$; otherwise $R_2 = R(b, n)$, where $m_1 = s_x$ (see Figs. \ref{Fig-LongPath1}(a)--(b)). Let $p\in V(R_1)$ and $q\in V(R_2)$ such that $p\thicksim q$, $p = (m_1, n)$, and $q = (m_1+1, n)$. First, by the algorithms of \cite{Hung17a} and \cite{Keshavarz19b}, we can construct a longest $(s, p)$-path $P_1$ in $R_1$ and a longest $(q, t)$-path $P_2$ in $R_2$. Then, $P = P_1\Rightarrow P_2$ forms a longest $(s, t)$-path of $O(m,n; k,l; a,b,c,d)$. Figs. \ref{Fig-LongPath1}(c) and (d) show the constructions of such a longest $(s,t)$-path. The size of constructed longest $(s, t)$-path equals to $\hat{L}(R_1, s, p) + \hat{L}(R_2, q, t) = m\times n-k\times l-t_x+s_x+1$ or $m\times n-k\times l-t_y-m+s_x+2$.

\begin{figure}[h]
\centering
\includegraphics[scale=0.85]{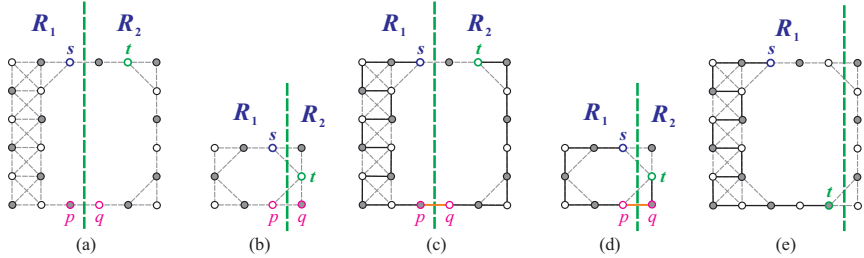}
\caption{(a) and (b) A vertical separation on $O(m,n; k,l; a,b,c,d)$ for (O1) and (O2) respectively, (c) and (d) a longest $(s, t)$-path in $O(m,n; k,l; a,b,c,d)$ for (a) and (b) respectively, and (e) a longest $(s, t)$-path in $O(m,n; k,l; a,b,c,d)$ for (O3), where bold lines indicate the constructed longest path between $s$ and $t$.}\label{Fig-LongPath1}
\end{figure}

Case 2: Condition (O3) holds. Then, by Lemma \ref{Lemma:O1-O3}, $\hat{U}(O(m,n; k,l; a,b,c,d), s, t) = \max\{\hat{L}(R_1, s, t), \hat{L}(R_2, s, t)\}$. Consider Figs. \ref{Fig-LongP1}(c) and \ref{Fig-LongP1}(d). Since $R_1$ and $R_2$ are $C$-shaped supergrid graphs, by the algorithm of \cite{Keshavarz19b} we can construct a longest path between $s$ and $t$ in $R_1$ or $R_2$. Fig. \ref{Fig-LongPath1}(e) depicts such a construction.

Case 3: Condition (F11) holds. Consider Fig. \ref{Fig-LongP2}(a). Then, by Lemma \ref{Lemma:F11-F13}, $\hat{U}(O(m,n; k,l; a,b,c,d), s, t) = \ell$, where $\ell = \max\{m\times n - k\times l - t_x+u_x+1, m\times n - k\times l -v_x+t_x+1\}$. There are the following two subcases:

\hspace{0.5cm}Case 3.1: $\ell = m\times n - k\times l - v_x+t_x+1$. We make two vertical and one horizontal separations on $O(m,n; k,l; a,b,c,d)$ to obtain three disjoint supergrid subgraphs $R_1 = L(m,n; k+b, c+l)$, $R_2 = R(k-1, c)$, and $R_3 = L(b+1,c+l; 1,l)$ (see Fig. \ref{Fig-LongPath2}(a)).
Let $p\in V(R_1)$ and $q\in V(R_2)$ such that $p\thicksim q$, $q = (a+1, 1)$, and $p = (a, 1)$ if $s\neq (a, 1)$; otherwise $p = (a, 2)$. Consider $(R_1, s, p)$. It is easy to check that $(R_1, s, p)$ does not satisfy conditions (F1)--(F3). By the algorithms of \cite{Hung17a} and \cite{Keshavarz19a}, we can construct a Hamiltonian $(s, p)$-path $P_1$ in $R_1$ and a longest $(q, t)$-path $P_2$ in $R_2$. Note that by the algorithm in \cite{Keshavarz19a} we can construct $P_1$ so that its one edge is placed to face $R_3$. Then, $P_{12} = P_1 \Rightarrow P_2$ forms a longest $(s, t)$-path of $R_1\cup R_2$, as depicted in Fig. \ref{Fig-LongPath2}(b). By Theorem \ref{HC-LCshaped}, $R_3$ contains a Hamiltonian cycle $HC_3$. By the algorithm in \cite{Keshavarz19a}, we can construct $HC_3$ such that its one flat face is faced to $R_1$. Then, there exist two edges $e_1\in P_{12}$ and $e_3 \in HC_3$ and such that $e_1 \thickapprox e_3$ (see Fig. \ref{Fig-LongPath2}(c)). By Statement (2) of Proposition \ref{Pro_Obs}, $P_{12}$ and $HC_3$ can be combined into a longest $(s, t)$-path $P$ of $O(m,n; k,l; a,b,c,d)$. The construction of a such longest path is depicted in Fig. \ref{Fig-LongPath2}(c). The size of constructed longest $(s, t)$-path equals to $\hat{L}(R_1, s, p) + \hat{L}(R_2, q, t) + |V(R_3)| = m\times n - k\times l -v_x+t_x+1$.

\begin{figure}[h]
\centering
\includegraphics[scale=0.85]{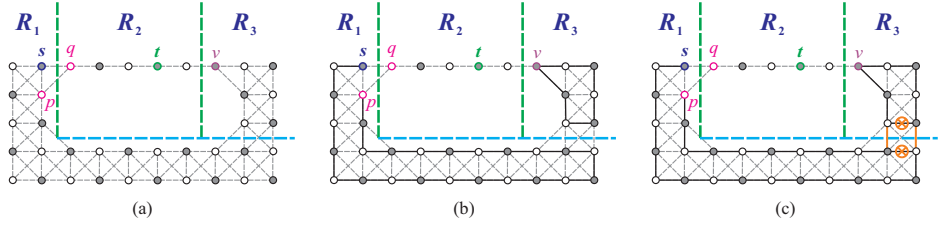}
\caption{(a) Two vertical and one horizontal separations on $O(m,n; k,l; a,b,c,d)$ for (F11) under that $\ell = m\times n - k\times l - v_x+t_x+1$, (b) a longest $(s, t)$-path in $R_1\cup R_2$ and a Hamiltonian cycle of $R_3$, and (c) a longest $(s, t)$-path in $O(m,n; k,l; a,b,c,d)$ for (a), where bold lines indicate the constructed longest $(s, t)$-path and $\otimes$ represents the destruction of an edge while constructing such a $(s, t)$-path.}\label{Fig-LongPath2}
\end{figure}

\hspace{0.5cm}Case 3.2: $\ell = m\times n - k\times l - t_x+u_x+1$. Consider the following subcases:

\hspace{1cm}Case 3.2.1: $s_x > a$. In this case, assume that $(1, 1)$ is the coordinates of vertex in upper-right corner in $O(m,n; k,l; a,b,c,d)$. Then, a longest $(s, t)$-path of $O(m,n; k,l; a,b,c,d)$ can be constructed by similar to Cases 3.1 (see Fig. \ref{Fig-LongPath3}(a)).

\begin{figure}[h]
\centering
\includegraphics[scale=0.85]{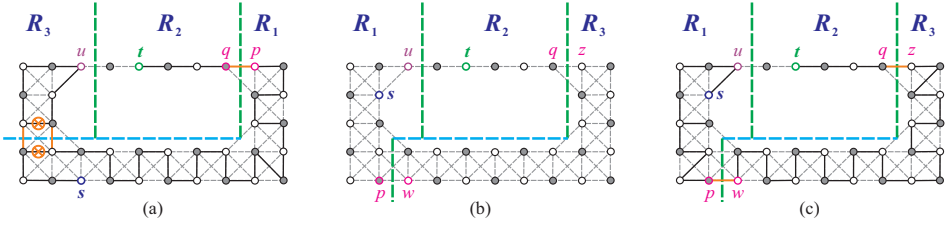}
\caption{(a) A longest $(s, t)$-path in $O(m,n; k,l; a,b,c,d)$ for (F11) under that $\ell = m\times n - k\times l - t_x+u_x+1$ and $s_x > a$, (b) three vertical and one horizontal separations on $O(m, n; k, l; a,b,c,d)$ for (F11) under that $\ell = m\times n - k\times l - t_x+u_x+1$ and $s_x \leqslant a$, and (c) a longest $(s, t)$-path in $O(m,n; k,l; a,b,c,d)$ for (b).}\label{Fig-LongPath3}
\end{figure}

\hspace{1cm}Case 3.2.2: $s_x\leqslant a$. We make three vertical and one horizontal separations on $O(m,n; k,l; a,b,c,d)$ to obtain three disjoint
supergrid subgraphs $R_1 = L(a+1,n; 1,n-c)$, $R_2 = R(k-1, c)$, and $R_3 = L(m-a,n; k,c+l)$ (see Fig. \ref{Fig-LongPath3}(b)). Let $p\in V(R_1)$, $q\in V(R_2)$, $w, z\in V(R_3)$ such that $p \thicksim w$, $q \thicksim z$, $q = (a+k, 1)$, $z = (a+k+1, 1)$, $w = (a+1, n)$, and $p = (a, n)$ if $s\neq (a, n)$; otherwise $p = (a, n-1)$. It is easy to verify that $(R_1, s, p)$ and $(R_3, w, z)$ do not satisfy conditions (F1)--(F3). By the algorithm of \cite{Keshavarz19a}, we can construct a Hamiltonian $(s, p)$-path $P_1$ and Hamiltonian $(w, z)$-path $P_3$ of $R_3$, respectively.
By the algorithm of \cite{Hung17a}, we can construct a longest $(q, t)$-path $P_2$ in $R_2$. Then, $P = P_1 \Rightarrow P_3\Rightarrow P_2$ forms a longest $(s, t)$-path of $O(m,n; k,l; a,b,c,d)$. The size of constructed longest $(s, t)$-path equals to $\hat{L}(R_1, s, p) + \hat{L}(R_3, w, z) + \hat{L}(R_2, q, t) = m\times n - k\times l - t_x-u_x+1$. The construction of such a longest $(s, t)$-path of $O(m,n; k,l; a,b,c,d)$ is depicted in Fig. \ref{Fig-LongPath3}(c).

Case 4: Case 2 of Condition (F13) holds. Consider Figs. \ref{Fig-LongP2}(b)--(c). Then, by Lemma \ref{Lemma:F11-F13}, $\hat{U}(O(m,n; k,l; a,b,c,d) , s, t) = \ell$, where $\ell = \max\{m\times n - k\times l - t_x+u_x+1, m\times n - k\times l - z_y-m+t_x+1\}$ (resp. $\ell = \max\{m\times n - k\times l - t_x+u_y+2, m\times n - k\times l - v_y+t_y+1\}$). There are the following two subcases:

\hspace{0.5cm}Case 4.1: $\ell = m\times n - k\times l - t_x+u_x+1$ (resp. $\ell = m\times n - k\times l -m+u_x-t_y+2)$. A longest $(s, t)$-path of $O(m,n; k,l; a,b,c,d)$ can be constructed by similar to Case 1, where $m_1 = a+1$, $q = (m_1+1, n)$, and $p = (m_1, n)$ (see Fig. \ref{Fig-LongPath4}(a)). The size of constructed longest $(s, t)$-path equals to $\hat{L}(R_1, s, p) + \hat{L}(R_2, q, t) = m\times n - k\times l - t_x+u_x+1$ (resp. $m\times n - k\times l -m-u_x-t_y+2)$. The construction of such a longest $(s, t)$-path of $O(m,n; k,l; a,b,c,d)$ is depicted in Fig. \ref{Fig-LongPath4}(b).

\begin{figure}[h]
\centering
\includegraphics[scale=0.85]{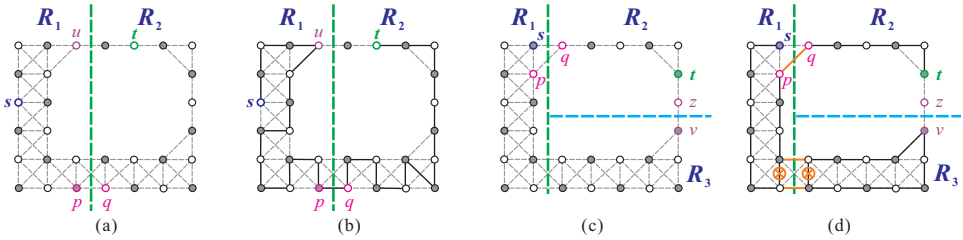}
\caption{(a) A vertical separation on $O(m,n; k,l; a,b,c,d)$ for case (2) of condition (F13) under that $\ell = m\times n - k\times l - t_x+u_x+1$, (b) a longest $(s, t)$-path in $O(m,n; k,l; a,b,c,d)$ for (a), (c) a vertical and a horizontal separations on $O(m,n; k,l; a,b,c,d)$ for case (2) of condition (F13) under that $\ell = m\times n - k\times l - z_y-m+t_x+1$, and (d) a longest $(s, t)$-path in $O(m,n; k,l; a,b,c,d)$ for (c).}\label{Fig-LongPath4}
\end{figure}

\hspace{0.5cm}Case 4.2: $\ell = m\times n - k\times l - z_y-m+t_x+1$ (resp. $\ell = m\times n - k\times l -v_y+t_y+1)$. A longest $(s, t)$-path of $O(m,n; k,l; a,b,c,d)$ can be constructed by similar to Case 3.1, where $R_1 = R(m_1, n)$, $R_2 = L(m-m_1,c+l-1; k,l-1)$, $R_3 = L(m-m_1,d+1; k,1)$, $m_1 = a$, $q = (m_1+1, 1)$, $p = (m_1, 1)$ if $s\neq (m_1, 1)$; otherwise $p = (m_1, 2)$ (see Fig. \ref{Fig-LongPath4}(c)). The size of constructed longest $(s, t)$-path equals to $\hat{L}(R_1, s, p) + \hat{L}(R_2, q, t) + |V(R_3)| = m\times n-k\times l - z_y-m+t_x+1$ (resp. $m\times n - k\times l -v_y+t_y+1)$. The construction of such a longest $(s, t)$-path of $O(m,n; k,l; a,b,c,d)$ is depicted in Fig. \ref{Fig-LongPath4}(d).

Case 5: Condition (F13) (cases 1.1 and 3), (F12), or (O4) holds. Then, by Lemma \ref{Lemma:F11-F13} or \ref{Lemma:O4}, $\hat{U}(O(m,n; k,l; a,b,c,d),s,t) = \max\{\hat{L}(R_1, s, u) + \hat{L}(R_2,q,t), \hat{L}(R_1, s, v) + \hat{L}(R_2, z, t)\}$. Consider Figs. \ref{Fig-LongP2}(d)--(g) and Fig. \ref{Fig-LongP3}. Since $R_1$ and $R_2$ are $C$-shaped supergrid graphs, first by the algorithm of \cite{Keshavarz19b} we can construct  a longest $(s, u)$-path $P_{11}$, a longest $(s, v)$-path $P_{12}$ in $R_1$, a longest $(q, t)$-path $P_{21}$, and a longest $(z, t)$-path $P_{22}$ in $R_2$. Then, $P = P_{11}\Rightarrow P_{21}$ or $P = P_{12}\Rightarrow P_{22}$ forms a longest $(s, t)$-path of $O(m,n; k,l; a,b,c,d)$.
\end{proof}

It follows from Theorem \ref{HP-Theorem-Oshaped} and Lemmas \ref{Lemma:O1-O3}--\ref{Lemma:long-Osupergrid} that the following theorem concludes the result.

\begin{thm}\label{LP-Theorem}
Let $O(m,n; k,l; a,b,c,d)$ be an $O$-shaped supergrid graph with vertices $s$ and $t$. Then, there exists a linear-time algorithm for finding the longest $(s, t)$-path of $O(m,n; k,l; a,b,c,d)$.
\end{thm}

The linear-time algorithm is formally presented as Algorithm \ref{TheHamiltonianPathAlgm}.

\begin{algorithm}[tb]
  \SetCommentSty{small}
  \LinesNumbered
  \SetNlSty{textmd}{}{.}

    \KwIn{An $O$-shaped supergrid graph $O(m,n;k,l;a,b,c,d)$ and two distinct vertices $s$ and $t$ in it.}
    \KwOut{The longest $(s, t)$-path.}

\textbf{if} $s_x, t_x\leqslant a$ \textbf{then} \textbf{output} $HP(O(m,n; k,l; a,b,c,d), s, t)$ constructed from Lemma \ref{HP-Oshaped1}; // $(O(m,n; k,l; a,b,c,d), s, t)$ does not satisfy the forbidden conditions (F1) and (F10);\\
\textbf{if} $s_x\leqslant a$ and $t_x\geqslant a+1$ \textbf{then} \textbf{output} $HP(O(m,n; k,l; a,b,c,d), s, t)$ constructed from Lemma \ref{HP-Oshaped2}; // $(O(m,n; k,l; a,b,c,d), s, t)$ does not satisfy the forbidden conditions (F1) and (F10)--(F13);\\
\textbf{if} $s_x, t_x\geqslant a+1$ \textbf{then} \textbf{output} $HP(O(m,n; k,l; a,b,c,d), s, t)$ constructed from Lemma \ref{HP-Oshaped3}; // $(O(m,n; k,l; a,b,c,d), s, t)$ does not satisfy the forbidden conditions (F1) and (F10)--(F13);\\
\textbf{if} $(O(m,n; k,l; a,b,c,d), s, t)$ satisfies one of the forbidden conditions (F1)and (F10)--(F13),  \textbf{then} \textbf{output} the longest $(s, t)$-path based on Lemma \ref{Lemma:long-Osupergrid}.\\
\caption{The longest $(s, t)$-path algorithm}
\label{TheHamiltonianPathAlgm}
\end{algorithm}

\section{Concluding Remarks}\label{Sec_Conclusion}
We gave necessary conditions for the existence of a Hamiltonian path in $O$-shaped supergrid graphs between any two given vertices. Then we showed that these necessary conditions are also sufficient by giving a linear-time algorithm to compute the Hamiltonian path between any two vertices. That is, $O$-shaped supergrid graphs are Hamiltonian connected except five forbidden conditions. We finally present a linear-time algorithm to compute the longest $(s, t)$-path of an $O$-shaped grid graph given any two vertices $s$ and $t$ when the forbidden conditions are satisfied. $O$-shaped supergrid graphs are a special kind of supergrid graphs with some holes. So, solving the Hamiltonian and the longest path problem for $O$-shaped supergrid graphs can be considered among the first attempts to solve the problems for more general cases of supergrid graphs. The Hamiltonian and  longest path problems are NP-complete for general supergrid graphs \cite{Hung15}. But it is still open for supergrid graphs with some holes or without hole. We would like to post it as an open problem to interested readers.

\section*{Acknowledgments}
This work is partly supported by the Ministry of Science and Technology, Taiwan under grant no. MOST 108-2221-E-324-012-MY2.

\end{document}